%% file: main_new.tex
\documentclass[a4paper,10pt,reqno,nonamelimits]{article}
\usepackage[utf8]{inputenc}
\usepackage{a4wide}
\usepackage{amsfonts,newlfont,amscd,amsthm,amsxtra,mathrsfs,nicefrac,pifont}
\usepackage{bold-extra}
\usepackage{booktabs}
\usepackage{longtable}
\usepackage{graphicx,url}
\usepackage[bookmarks=false,colorlinks=true,urlcolor=blue,linkcolor=blue,citecolor=blue]{hyperref}

\usepackage{mathabx}
\usepackage{dsfont}
\usepackage{cleveref}
\usepackage{todonotes}
\usepackage{authblk}
\crefname{theorem}{thrm.}{thrms.}
\crefname{enumi}{case}{cases}

\input{special_commands}

\title{A partial information decomposition for discrete and continuous variables}

\author[1, 2]{\small Kyle Schick-Poland\footnote{Corresponding author \href{mailto:kyle.schick-poland@uni-goettingen.de}{kyle.schick-poland@uni-goettingen.de}}}
\author[1]{\small Abdullah Makkeh\footnote{Corresponding author \href{mailto:abdullah.alimakkeh@uni-goettingen.de}{abdullah.alimakkeh@uni-goettingen.de}}}
\author[1, 3]{\small Aaron J. Gutknecht}
\author[2]{\small Patricia Wollstadt}
\author[4]{\small Anja Sturm}
\author[1]{\small Michael Wibral}

\affil[1]{\footnotesize Campus  Institute  for  Dynamics  of  Biological  Networks,  Georg-August-University,  Goettingen,  Germany}
\affil[2]{\footnotesize Honda Research Institute Europe, Offenbach am Main, Germany}
\affil[3]{\footnotesize MEG Unit, Brain ImagingCenter, Goethe University, Frankfurt, Germany}
\affil[4]{\footnotesize Institute for Mathematical Stochastics, Georg-August-University,  Goettingen,  Germany}

\date{\today}

\begin{document}

\maketitle
\begin{abstract}
 
Conceptually, partial information decomposition (PID) is concerned with separating the information contributions several sources hold about a certain target by decomposing the corresponding joint mutual information into contributions such as synergistic, redundant, or unique information. Despite PID conceptually being defined for any type of random variables, so far, PID could only be quantified for the joint mutual information of discrete systems. Recently, a quantification for PID in continuous settings for two or three source variables was introduced. Nonetheless, no ansatz has managed to both quantify PID for more than three variables and cover general measure-theoretic random variables, such as mixed discrete-continuous, or continuous random variables yet. In this work we will propose an information quantity, defining the terms of a PID, which is well-defined for any number or type of source or target random variable. This proposed quantity is tightly related to a recently developed local shared information quantity for discrete random variables based on the idea of shared exclusions. Further, we prove that this newly proposed information-measure fulfills various desirable properties, such as satisfying a set of local PID axioms, invariance under invertible transformations, differentiability with respect to the underlying probability density, and admitting a target chain rule.

\end{abstract}

\section{Introduction} \label{section_intro}
Partial information decomposition (PID) aims at describing the structure of the precise mechanisms, in which multiple random variables $S_1, \ldots,S_n$, called \emph{sources}, provide information about a specific random variable $T$, called \emph{target}\cite{williams2010nonnegative}. This is achieved by decomposing the joint mutual information $I(T:S_1, \ldots,S_n)$ into information contributions that fulfill certain requirements of part-whole relationships (mereology) as shown in \cite{gutknecht2020bits}. For instance, a PID of $I(T: S_1, S_2)$ formalizes the intuitive notion that two sources random variables may provide information about a target uniquely -- such that this information is not available from the other source variable, redundantly -- such that this type of this information about the target variable can be obtained from either source variable, or synergistically -- such that the information about the target can only be accessed when considering both source variables together. Generally, PID details the way that information about the target variable is provided by multiple source variables and collections thereof.

The PID problem was first described in the seminal work of Williams and Beer~\cite{williams2010nonnegative} who laid out that classic information theory must be extended by additional assumptions or axioms to make the problem well posed. 
In 2013, Bertschinger and colleagues have demonstrated that certain intuitive desiderata of a PID measure are incompatible~\cite{bertschinger2013shared}. Thus a whole landscape of different choices of additional assumptions has been suggested---drawing on motivations from fields outside of information theory and often tailored to specific perspectives. 
Listing some of the most influential approaches, in 2012 Harder et al. quantify a PID by projecting probability distributions onto conditional distributions in an optimization approach~\cite{harder_polani_geometric_measure}.
On the other hand, Bertschwinger et al. derive a unique information quantity invoking the intuition that unique information of a source about a target is determined by the marginal distribution of only those two random variables, leading to a minimization problem over distributions with fixed marginals~\cite{bertschinger_unique}.
In 2017, Ince defined a PID quantity motivated from agents performing non-cooperative games~\cite{ince_game_theoretical_approach}. Taking a more classical information-theoretic perspective, Rosas et al. have quantified an approach for a quantity intuitively resembling a synergistic information based on an optimization problem with respect to information channels, leading to a PID-like decomposition~\cite{rosas_synergy_channel_measure}. While many approaches, including the ones mentioned before, attempt to quantify a PID via averaged information measures, the work from Finn and Lizier in 2013 argues for the relevance of \emph{locally} defined PIDs~\cite{finn_lizier_local_measures}, simultaneously introducing a set of local axioms parallel to the axiomatic approach from Williams and Beer in~\cite{williams2010nonnegative}. Extending these arguments, in 2018 Finn and Lizier investigated the lattice structure of PID, imposing a pointwise separation in so called specification and ambiguity lattices. 

Apart from not being defined locally, none of these previously suggested measures has been shown to be differentiable with respect to the underlying probability distributions---another desirable property from the point of view of information theory. Recently, Makkeh et al.\; recently have proposed a local differentiable information-measure, yielding a partial information decomposition that is solely based on information-theoretic principles and defined for an arbitrary number of source variables~\cite{makkeh2021isx}.
Due to the largely different sets of assumptions it is needed to say that all the aforementioned approaches result in differently valued PID terms. 

However, these previous attempts have been proposed exclusively for discrete random variables. While there have been acknowledgeable advances towards quantifying a continuous PID \cite{barrett_continuous_gaussians,ari_continuous_variables}, these ansatzes could not capture the complexity of a PID beyond two source random variables.

In this work we show that the specific framework from \cite{makkeh2021isx} of partial information decomposition can be extended to systems of continuous variables, and mixed systems of discrete and continuous variables. In particular, we demonstrate that our local information-measure is well-defined measure-theoretically, invariant under invertible transformations, fulfills a target chain rule for composite target random variables, remains differentiable and recovers the discrete definition under consideration of the Dirac-measure. Further we define a corresponding global (averaged) information-measure, that will inherit the same properties as the local one.

This document will be structured as follows. In \cref{section_pid} we shortly explain the conceptual structure of a PID, and reintroduce the information-measure from \cite{makkeh2021isx} that will be generalized to a measure-theoretic setting throughout this work.

Next, some mathematical background has to be recalled in \cref{section_mathematical_background}, stating a number of famous theorems from measure theory and probability theory.

In \cref{section_continuation_PID}, we will be applying the rigor of \cref{section_mathematical_background} to the aforementioned recently introduced PID framework described in \cref{section_pid}, generalizing the expressions to a measure-theoretic setting. In doing so, we have automatically covered the cases of continuous, and mixed random variables. Throughout \cref{subsubsection_recovering_discrete_case,subsubsection_self_redundancy_monotonicity_symmetry,subsection_properties_subsubsection_invariance,subsubsection_differentiability_of_isx,subsubsection_properties_target_chain_rule}, the resulting newly defined quantity will then be demonstrated to recover the discrete PID quantity for a suitable underlying measure, fulfill a version of the pointwise PID axioms as in \cite{finn2018probability}, be invariant under invertible transformations, as well as differentiable, and admit to a target chain rule, respectively.

\section{Partial information decomposition} \label{section_pid}
\subsection{What is partial information decomposition?}
We here briefly present the essentials of PID theory (for more details see~\cite{gutknecht2020bits,williams2010nonnegative}). The starting point of partial information decomposition is given by jointly distributed random variables $S_1,\ldots,S_n, T$ where the $S_i$ are called the "information sources" and T is called the "target". As long as the sources are not statistically independent of the target, they will provide non-zero joint mutual information $I(T:S_1,\ldots,S_n)$ about the target quantifying the degree of dependence \cite{cover1999elements}. The goal of PID is now to answer the question of "who among the sources knows what about the target?". In the case of two sources $S_1$ and $S_2$ this should lead to a decomposition of the joint mutual information $I(T:S_1,S_2)$ into four components (see Figure \cref{fig:2_sources_pid}): the information \textit{uniquely} carried by the first source, the information \textit{uniquely} carried by the second source, the information \textit{redundantly} carried by both sources, and the information that only arises \textit{synergistically} by combining information from both sources. At the same time, the mutual information terms associated with the individual sources, $I(T:S_1)$ and $I(T:S_2)$, should decompose into a unique and a redundant part each. 

\begin{figure}[ht] 
	\centering
	\includegraphics[width=0.5\textwidth]{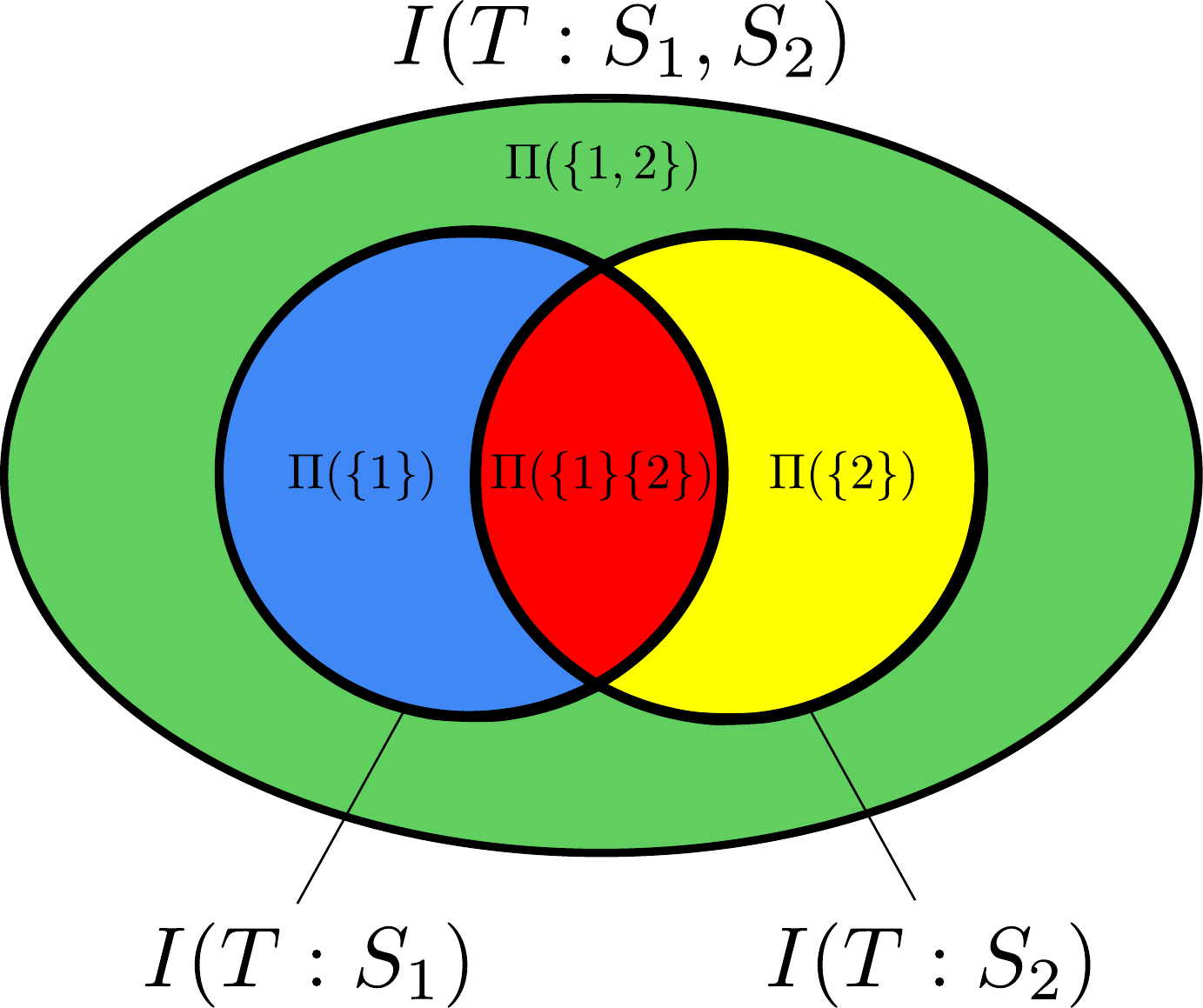}
	\caption{Illustration of partial information decomposition in the case of two sources. The joint mutual information (entire oval shape) decomposes into four pieces: unique information of source 1 (blue patch, $\Pi(\{1\})$), unique information of source 2 (yellow patch, $\Pi(\{2\})$), redundant information (red patch, $\Pi(\{1\}\{2\})$), and synergistic information (green patch, $\Pi(\{1,2\})$). The individual mutual information terms (black circles) each decompose into a unique and a redundant part.}
	\label{fig:2_sources_pid}
\end{figure}

This idea can be generalized to the case of $n$ sources by realizing that the four components are characterized by their distinctive parthood relationships to the information provided by the different possible \textit{subsets} of sources: the information uniquely carried by the first source is only part of $I(T:S_1)$ but not part of $I(T:S_2)$ and the other way around for the information uniquely carried by the second source. The information redundantly carried by both sources is part of both $I(T:S_1)$ and $I(T:S_2)$. And the synergistic information is only part of the joint mutual information $I(T:S_1, S_2)$ but not part of either $I(T:S_1)$ or $I(T:S_2)$. In other words, PID can be thought of as a decomposition of the joint mutual information into components such that there is one component per possible arrangement of parthood relationships to the mutual information provided by the different subsets of sources. The idea of a "possible arrangement of parthood relationships" can be formalized with the notion of a \textit{parthood distribution}:
\begin{definition}
	A parthood distribution (in the context of n source variables) is any function $f:\mathcal{P}\left(\{1,\ldots,n\}\right) \rightarrow \{0,1\}$ such that
	\begin{enumerate}
		\item $f(\emptyset) = 0$ ("There is no information in the empty set")
		\item $f(\{1,\ldots,n\}) = 1$ ("All information is in the full set")
		\item For any two collections of source indices $\mathbf{a}$, $\mathbf{b}$: If $\mathbf{b} \supseteq \mathbf{a}$, then $ f(\mathbf{a}) = 1 \Rightarrow f(\mathbf{b}) = 1$ (Monotonicity)
	\end{enumerate}
\end{definition} 
Intuitively, any "piece" of the joint mutual information is associated with a parthood distribution that describes for each subcollection of sources $\mathbf{a}\subseteq \{1,\ldots,n\}$ whether or not the piece is part of the mutual information $I(T:(S_i)_{i\in\mathbf{a}})$ provided by the subcollection $\mathbf{a}$. In the former case $f(\mathbf{a})=1$ and in the latter  $f(\mathbf{a})=0$. However, not all arrangements of these parthood relationships are possible which is why parthood distributions have to satisfy the three constraints in the above definition. First, no piece of information should be part of the information provided by the empty set of sources, because if we do not know any source we do not obtain any information that the sources may carry about the target. Secondly, since we are talking about pieces of the joint mutual information any such piece is trivially part of the joint mutual information. Finally, if a piece of information is part of the information provided by a subcollection $\mathbf{a}$, then it also has to be part of the information provided by any superset of $\mathbf{a}$. 

Thus, the idea underlying PID is to decompose the joint mutual information into one component per parthood distribution such that the component $\Pi(f)$ associated with parthood distribution $f_\Pi$ is exactly that part of $I(T:S_1,\ldots,S_n)$ that stands in the parthood relations to the different $I(T:(S_i)_{i\in \mathbf{a}})$ specified by $f_\Pi$. As an example consider the case of two sources and the parthood distribution $f$ such that 
\begin{equation}
f(\emptyset)= 0,\hspace{0.1cm} f(\{1\})=1, \hspace{0.1cm}  f(\{2\})=0, \hspace{0.1cm}  f(\{1,2\})=1 \; .
\end{equation}
Then the component $\Pi(f)$ is exactly (i.e. all and only) the information that is part of the mutual information provided by source 1 but not part of the information provided by source 2 (and of course trivially it is part of the information carried by the full set of sources). In other words, $\Pi(f)$ is the information carried uniquely by source 1. 

Until now we have discussed only the qualitative relationships between the components $\Pi$ and the different mutual information terms, i.e. we have discussed which mutual information terms should be made up of which components. Now, on the quantitative side PID generally assumes an \textit{additive} relationship between these quantities: PID is a decomposition of the mutual information into components that simply \textit{add up} to the mutual information. Based on these ideas we may now put forth the following minimal definition of a partial information decomposition:
\begin{definition}[Minimally consistent PID] \label{def:mcpid_sum}	Let $S_1,\ldots,S_n, T$ be jointly distributed random variables with joint distribution $\mathbb{P}_J$ and let $\mathcal{F}_n$ be the set of parthood distributions in the context of n source variables. A minimally consistent partial information decomposition of the joint mutual information provided by the sources $S_1,\ldots,S_n$ about the target $T$ is any function $\Pi_{\mathbb{P}_j}:\mathcal{F}_n \rightarrow \mathbb{R}$, determined by $\mathbb{P}_J$, that satisfies
	\begin{equation}\label{eq:consistency_equation}
	I_{\mathbb{P}_J}(T:(S_i)_{i \in \mathbf{a}}) = \sum\limits_{\substack{f \in \mathcal{F}_n \\ f(\mathbf{a}) = 1}} \Pi_{\mathbb{P}_J}(f)
	\end{equation}
	for all $\mathbf{a} \subseteq \{1,\ldots,n\}$. The subscripts $\mathbb{P}_J$ indicate that both the mutual information and the information atoms are functions of the underlying joint distribution.
\end{definition}

\subsection{PID and redundant information}\label{sec:pid-redundant-info}
In principle, a PID could be specified by defining all $\Pi(f)$ directly making sure condition \cref{eq:consistency_equation} is satisfied. However, a more convenient option is to shift the problem to specifying a measure of the redundant information $I_\cap(T:\mathbf{a}_1; \ldots; \mathbf{a}_m)$ provided by collections $\mathbf{a}_1,\ldots, \mathbf{a}_m$ about the target. Note that the redundant information \textit{component} in the case of two sources (red patch in Figure \cref{fig:2_sources_pid} above) is only a special case of redundant information. In general we should expect redundant information to consist of multiple components. For instance, in the context of three sources, the redundant information of source 1 and 2 will consist of a component representing the information shared by all three sources and a component representing the information shared by sources 1 and 2 but not by source 3.

Just like the components $\Pi$, a measure of redundant information is not defined within classical information theory. However, simply based on the meaning of the word "redundant" it seems reasonable to demand that redundant information should stand in a very specific relationship to the components $\Pi(f)$ introduced in the previous section. Specifically, the redundant information $I_\cap(T:\mathbf{a}_1; \ldots; \mathbf{a}_m)$ should consist of all components $\Pi(f)$ that are part of $I(T:(S_i)_{i \in \mathbf{a}_j})$ for \textit{all} $j=1,\ldots,m$ leading to the equation
\begin{equation}\label{eq:rel_red_atom}
I_\cap(T:\mathbf{a}_1;\ldots; \mathbf{a}_m) = \sum_{\substack{f \in \mathcal{F}_n\\\forall j f(\mathbf{a}_j)=1 }} \Pi(f) \; .
\end{equation}
This equation is to be understood as a constraint on the functions $I_\cap$ and $\Pi$ which are both yet to be defined. It can be shown that equation \cref{eq:rel_red_atom} can be inverted to obtain a unique solution for the $\Pi(f)$ once a measure of redundant information $I_\cap$ is specified \cite{gutknecht2020bits}. However, $I_\cap$ cannot be chosen abritrarily since the equation implies the following properties: 
\begin{enumerate}
	\item $I_\cap(T:\mathbf{a}_1; \ldots;\mathbf{a}_m) = I_\cap(T:\mathbf{a}_{\kappa(1)}; \ldots; \mathbf{a}_{\kappa(m)})$ for any permutation $\kappa$ (\textbf{symmetry})
	\item If $\mathbf{a}_i \supseteq \mathbf{a}_j$ for $i\neq j$, then ${I_\cap(T:\mathbf{a}_1; \ldots; \mathbf{a}_m)}$ $=$ ${I_\cap(T:\mathbf{a}_1; \ldots; \mathbf{a}_{i-1}; \mathbf{a}_{i+1}; \ldots; \mathbf{a}_m)}$ \textbf{ (invariance under superset removal / addition)}
	\item $I_\cap(T:\mathbf{a}) = I(T:\mathbf{a})$ (\textbf{self-redundancy})
\end{enumerate}
The first two properties imply that the definition of a redundancy measure can be restricted to antichains of $(\mathcal{P}(\{1,\ldots,n\}),\subseteq)$, i.e. to a set of subsets $\{\mathbf{a}_1,\ldots,\mathbf{a}_m\}$ such that no $\mathbf{a}_i$ is a subset of some $\mathbf{a}_j$ for $i\neq j$. The third condition ensures that the consistency equation $\cref{eq:consistency_equation}$ is satisfied if the components $\Pi$ are determined by inversion of \cref{eq:rel_red_atom}.

A measure of redundant information satisfying all three properties, is the $I_\cap^{\sx}$ measure that was introduced in \cite{makkeh2021isx} for the case of discrete random variables. In the following, we will briefly explain this measure. The key aim of this paper is to provide a continuous extension of $I_\cap^{\sx}$, and hence, a continuous PID.

\subsection{The $i_\cap^{\sx}$ measure of redundant information} \label{section_PID_isx_measure}
The key idea behind the $I_\cap^{\sx}$ measure is to first define a \textit{pointwise} measure of the redundant information that a particular realization $s_1,\ldots,s_n$ of the source variables carries about a particular realization t of the target variable. This measure can then be averaged to obtain the global measure $I_\cap^{\sx}$. The local measure, denoted $i_\cap^{\sx}$, can be motivated in two distinct ways: firstly, based on the idea that redundant information is related to events being excluded (rendered impossible) redundantly by each source realization, and secondly, in terms of a particular logical statement about the source realizations that captures precisely the information shared by all of them
(for full details see \cite{makkeh2021isx}\cite{gutknecht2020bits}).

The basic setup for the $i_\cap^{\sx}$ measure is a  probability space $(\Omega, \mathfrak{A}, \mathbb{P})$ and discrete and finite random variables $S_1, ..., S_n, T$ be on that space, i.e. 
\begin{align*}
\small
&S_i: \Omega \rightarrow E_{S_i}, \hspace{0.3cm }(\mathfrak{A},\mathscr{P}( E_{S_i}))-\text{measurable} \\
&T: \Omega \rightarrow E_{T}, \hspace{0.3cm }(\mathfrak{A},\mathscr{P}(E_{T}))-\text{measurable},
\end{align*}
where $E_{S_i}$ and $E_T$ are the finite alphabets of the corresponding random variables and $\mathscr{P}(E_{S_i})$ and $\mathscr{P}( E_{T}) $ are the power sets of these alphabets. Now let the set of antichains of  $(\mathcal{P}(\{1,\ldots,n\}),\subseteq)$ be denoted by $\mathscr{A}_n$.

Then the \emph{local shared information} $i_\cap^{\sx}(t:\alpha)$ of an antichain $\alpha = \{\mathbf{a}_1 , \dots, \mathbf{a}_m\}$ (representing a set of collections of source realizations) about the target realization $t \in E_T$ is  defined in terms of the original probability measure $\mathbb{P}$ as a function $i_\cap^{\sx}: E_T \times \mathscr{A}_n \rightarrow \mathbb{R}$ with
\begin{equation} \label{eqn_isx_discrete_case}
\small
i_\cap^{\sx}(t:\alpha) = i_\cap^{\sx} (t:\mathbf{a}_1; \dots; \mathbf{a}_m) := \log_2\frac{\mathbb{P}\left(\mathfrak{t}|\bigcup_{i=1}^m \mathfrak{a}_i \right)}{\mathbb{P}(\mathfrak{t})} \; ,
\end{equation}
where $\mathfrak{a}_i = \bigcap_{j \in \mathbf{a}_i} \{S_j=s_j\} $ and $\mathfrak{t}= \{T=t\}$. A special case of this quantity  is the local shared information of a \emph{complete} sequence of source realizations $(s_1, \dots, s_n)$ about the target realization $t$. This is obtained by setting $\mathbf{a}_i = \{i\}$ and $m=n$:
\begin{equation*}
\small
i_\cap^{\sx} (t:\{1\};\dots;\{n\}) = \log_2\frac{\mathbb{P}\left(\mathfrak{t}|\bigcup_{i=1}^n  \{S_i=s_i\}\right)}{\mathbb{P}(\mathfrak{t})} \; .
\end{equation*}

One could obtain the same form of \cref{eqn_isx_discrete_case} including the probability mass instead of the probability measure, by definition of an auxiliary random variable $\mathcal{W}_\alpha = \left( \bigvee_{i=1}^m \bigwedge_{j \in \textbf{a}_i} S_j = s_j \right)$, where this auxiliary ``classifier'' random variable is one if the statement in the brackets is true, and zero otherwise. As can be seen in \cite{makkeh2021isx}, this auxiliary classifier random variable then provides exactly the same information as all the members of the antichain do. Conceptually, the truth of $\mathcal{W}_\alpha$ can be obtained from any of the members of the antichain, and thus captures the information which has been \emph{shared}, or is given \emph{redundantly}, by all these members. Then $\mathcal{W}_\alpha$ allows for rewriting \cref{eqn_isx_discrete_case} as
\begin{equation} \label{eqn_isx_discrete_case_auxiliary_random_variable}
\small
i_\cap^{\sx}(t:\alpha) = \log_2\frac{p\left(t|\mathcal{W}_\alpha = \mathrm{true} \right)}{p(t)} \;.
\end{equation}
This sort of form particularly becomes important when generalizing the equations given here, as the probability mass function is the density function in the discrete case.

Rewriting $i_\cap^{\sx}$ as in \cref{eqn_isx_discrete_case} allows us to decompose it into the difference of two positive parts:
\begin{equation*}
\small
\begin{split}
i_\cap^{\sx} (t:\mathbf{a}_1; \dots; \mathbf{a}_m)    &= \log_2\frac{\mathbb{P}\left(\mathfrak{t} \cap \bigcup_{i=1}^m \mathfrak{a}_i \right)}{\mathbb{P}(\mathfrak{t}) \mathbb{P}\left(\bigcup_{i=1}^m \mathfrak{a}_i \right)} = \log_2\frac{1}{ \mathbb{P}\left(\bigcup_{i=1}^m \mathfrak{a}_i \right)} -  \log_2\frac{\mathbb{P}(\mathfrak{t})}{\mathbb{P}\left(\mathfrak{t} \cap \bigcup_{i=1}^m \mathfrak{a}_i \right)} \; ,
\end{split}
\end{equation*}
We call 
\begin{equation*}
\small
i_\cap^{\sx+}(t : \mathbf{a}_1; \dots; \mathbf{a}_m) := \log_2 \frac{1}{ \mathbb{P}\left(\bigcup_{i=1}^m \mathfrak{a}_i \right)}
\end{equation*}
the \emph{informative} local shared information and
\begin{equation*}
\small
i_\cap^{\sx-}(t : \mathbf{a}_1; \dots; \mathbf{a}_m) := \log_2\frac{\mathbb{P}(\mathfrak{t})}{\mathbb{P}\left(\mathfrak{t} \cap \bigcup_{i=1}^m \mathfrak{a}_i \right)}
\end{equation*}
the \emph{misinformative} local shared information. These terms can be shown to induce non-negative components $\pi^{\sx+}$ and $\pi^{\sx-}$ \cite{makkeh2021isx}. Such a separation is desirable because $I_\cap^{\sx}$ as well as the implied decomposition $\Pi^{\sx}$ may have negative values for some antichains $\alpha$ / parthood distributions $f$. However, from an information theoretic standpoint it is important to have an explanation for this fact. After all, (global) mutual information itself is always non-negative. The idea is that negative values of $I_\cap^{\sx}$ and $\Pi^{\sx}$ can be interpreted as misinformation, i.e. in terms of situations in which the source realizations are misleading with respect to the actual target realization. This explanation is then supported by the fact that if informative and misinformative contributions are separated, then everything is non-negative again (for more details on informative/misinformative components see also \cite{finn2018probability}).

In the next section we will provide the necessary mathematical preliminaries for constructing a continuous extension of $i_\cap^{\sx}$. Due to the conditional probability involved in its definition, this will in particular include a discussion of the regular conditional probability.

\section{Preliminaries for a continuous extension of $i_\cap^{\sx}$} \label{section_mathematical_background}

This section considers itself with laying the groundwork for a theory of regular conditional probability. This is necessary since equations of the form \cref{eqn_isx_discrete_case} condition on single point events. In elementary probability, conditioning on such events is impossible as they have an associated measure of 0. Regular conditional probability, however, provides a rigorous framework for disintegrating the image measure while conserving the properties needed for a suitable candidate for generalization.
 
We choose the underlying space in a way that all probability measures, and all random variables, admit of a regular conditional probability. These spaces we will call \emph{Radon Borel probability spaces}. This definition offers a short notation of the complex objects seen throughout most literature.

\begin{definition}[Radon Borel Probability Spaces] \label{def_radon_borel_probability_spaces}
Let $(\Omega, \tau_\Omega)$ be a Radon topological space with respect to the Borel $\sigma$-algebra $\mathcal{B}(\Omega)$ generated by $\tau_\Omega$. Then we call $\left((\Omega, \tau_\Omega), \mathcal{B}(\Omega) \right)$ a \emph{Radon Borel measurable space}. 

Further consider a complete probability measure $\mathbb{P}$ on $\mathcal{B}(\Omega)$. Then we call the probability space $\left((\Omega, \tau_\Omega), \mathcal{B}(\Omega), \mathbb{P} \right)$ a \emph{Radon Borel probability space}.\footnote{Note that this definition makes every Radon Borel probability space a standard probability space.}
\end{definition}

In order to define a regular conditional probability, we need to disintegrate the image probability measure with respect to an outcome of the associated random variable. This idea was first utilized by von Neumann for probability measures in order to investigate the connection between measure theory and ergodic theory\cite{disintegration_theorem_neumann}. Therefore we will state an adapted version of Lebesgue's disintegration theorem:

\begin{theorem}[Disintegration theorem] \label{thrm_disintegration}
Let $\left((\Omega, \tau_\Omega), \mathcal{B}(\Omega), \mathbb{P} \right), \left((E, \tau_E), \mathcal{B}(E) \right)$ be a Radon Borel probability space, and a Radon Borel measurable space, respectively.
Extend $\left((E, \tau_E), \mathcal{B}(E) \right)$ to a Radon Borel probability space by the pushforward of $\mathbb{P}$ by a random variable $S: (\Omega, \mathcal{B}(\Omega), \mathbb{P})  \rightarrow (E, \mathcal{B}(E), \mathbb{P}^S := \mathbb{P} \circ S^{-1})$. 

Then there exists a $\mathbb{P}^S$-almost everywhere unique family of probability measures $\{\mathbb{P}_{x}\}_{x \in E}$ fulfilling:

\begin{itemize}
    \item $x \mapsto \mathbb{P}_{x}(B)$ is Borel measurable $\forall B \in \mathcal{B}(\Omega)$
    \item $\mathbb{P}_{x}(\Omega \setminus S^{-1}(x)) = 0$, i.e.~the disintegrated measure is not supported outside of the preimage of $x$ under $S$
    \item for all $\mathcal{B}(\Omega)$-measurable functions $f: \Omega \rightarrow \mathbb{R}$ there holds:
    \begin{align} \label{eqn_thrm_disintegration_defining_equality}
        \int\limits_{\Omega} f(\omega) d\mathbb{P}(\omega) = \int\limits_{E} d\mathbb{P}^S(x) \int\limits_{S^{-1}(x)} f(\omega) d\mathbb{P}_{x}(\omega) \; .
    \end{align} 
\end{itemize}
The family of probabilities, considered as a map $\nu: E \times \mathcal{B}(\Omega) \rightarrow [0,1], (x, B) \mapsto \mathbb{P}_{x}(B)$ is also called a \emph{transition probability}.
\end{theorem}

Note that there is a number of different formulations of this theorem, with the corresponding proofs varying just as much \cite{disintegration_theorem_pachl, conditional_expected_value, disintegration_theorem_chang_pollard}. Some explicitly require the additional property of the Borel algebras to be countably generated. This property is,  however, included in the version here, if one recalls that the separability of Radon spaces implies the corresponding Borel algebras to be countably generated. Furthermore, since these disintegrated measures are still probability measures over a Radon space, they are automatically Radon measures, and hence regular.

The disintegration theorem allows us to define a regular conditional probability \cite{regular_conditional_probability}.

\begin{theorem}[Regular conditional probability] \label{thrm_regular_conditional_probability}
Let $(\Omega, \mathcal{A}, \mathbb{P})$ and $(E, \mathcal{E})$ be Radon Borel probability and Radon Borel measurable spaces, respectively, with a surjective random variable $S$ mapping between them. Let further $\mathcal{E}$ be generated via $S$, i.e. $\mathcal{E} = \{S(A) | A \in \mathcal{A} \}$, and consider a $\sigma$-sub-algebra $\mathcal{F} \subset \mathcal{E}$. Observe that $\mathcal{F}$ is generated by a $\sigma$-sub-algebra $\mathcal{C} \subset \mathcal{A}$, that is $\mathcal{C} = S^{-1}\left(\mathcal{F}\right)$. 

Then there exists a transition probability $\nu_{\mathcal{F}}: E \times \mathcal{A} \rightarrow [0,1]$, satisfying 
\begin{align} \label{eqn_regular_conditional_probability_both_spaces}
    \mathbb{P}(A \cap S^{-1}(F)) = \int\limits_F \nu_{\mathcal{F}}(x, A) d(\mathbb{P} \circ S^{-1})(x)
\end{align}
for any $A \in \mathcal{A}$ and $F \in \mathcal{F}$. This transition probability is called a \emph{regular conditional probability, conditioned on} $\mathcal{F}$. For $C \in \mathcal{C}$, this can be pulled back to an integral on $\Omega$ as follows:
\begin{align}\label{eqn_regular_conditional_probability_sample_space}
    \mathbb{P}(A \cap C) = \int\limits_C \nu_{\mathcal{C}}(\omega, A) d\mathbb{P}(\omega) \; .
\end{align}
Here $\nu_{\mathcal{C}}:\Omega \times \mathcal{A} \rightarrow [0,1]$ is a regular conditional probabilty conditioned on $\mathcal{C}$,  fulfilling $\nu_{\mathcal{C}}(\omega, A) = \nu_{\mathcal{F}}(S(\omega), A)$.
\end{theorem}

Another perspective, which at the same time assures the existence of the regular conditional probability, is given by a theorem by Radon \cite{radon_nikodym_throrem_radon}, which was later generalized to its present form by Nikod\'{y}m \cite{radon_nikodym_throrem_nikodym}. This result is particularly important, since it allows for a tractable form of the regular conditional probability, from which several properties, of the regular conditional probability itself, as well as of the composed quantities presented in \cref{section_continuation_PID}, can be derived. These properties, found in \cref{section_continuation_PID_subsection_properties} are crucial, specifically for possible applications e.g.~deriving learning rules for neural networks and distributed computing.

We start by defining the notion of absolute continuity that is one of the necessary conditions for the Radon-Nikod\'{y}m theorem. 

\begin{definition}[{Absolute Continuity~\cite{conditional_expected_value}}] \label{def_absolute_continuity}
Let $\mu, \nu$ be $\sigma$-finite and finite positive measures on a measure space $(\Omega, \mathcal{A})$. Then $\nu$ is said to be \emph{absolutely continuous} with respect to $\mu$, denoted by $\nu \ll \mu$, iff $\mu(A)=0 \implies \nu(A)=0 \; \forall A \in \mathcal{A}$.
\end{definition}

\begin{theorem}[{Radon-Nikod\'{y}m Theorem~\cite{radon_nikodym_throrem_nikodym}}] \label{thrm_radon_nikodym}
Let the setting be as in \cref{def_absolute_continuity}.

Suppose that $\nu \ll \mu$, then there exists an $\mathcal{A}$-measurable $\mu$-almost everywhere unique function $f_{RN} \in L^1(\Omega)$, $f_{RN}: \Omega \rightarrow \mathbb{R}^+$, such that for all $E \in \mathcal{A}$
\begin{align*}
\nu(E) = \int_E f_{RN}(x) d\mu(x) \; .
\end{align*}
The function $f_{RN}$ is called Radon-Nikod\'{y}m derivative, commonly denoted by $\frac{d\nu}{d\mu}$.
\end{theorem}

Hence we can now rewrite the regular conditional probability as such a derivative via an auxiliary measure $\eta_A$.

With the statement from \cref{thrm_radon_nikodym} and using the same notation as in \cref{eqn_regular_conditional_probability_sample_space}, define a new measure for each $A \in \mathcal{A}$ by $\eta_A(C):= \mathbb{P}(A \cap C)$ for a set $A \in \mathcal{A}$ such that any regular conditional probability conditioned on some $\mathcal{C} \subset \mathcal{A}$, fulfills
\begin{align*}
    \eta_A (C) = \int\limits_C \nu_{\mathcal{C}}(\omega, A) d\mathbb{P}(\omega) \; \forall C \in \mathcal{C} \; .
\end{align*}

With $\mathbb{P}$ and $\eta_A$ being $\sigma$-finite positive measures for all $A \subset \Omega$ on $(\Omega, \mathcal{A}, \mathbb{P})$, and clearly $\eta_A \ll \mathbb{P}$ holding for all $A$, the corresponding Radon-Nikod\'{y}m derivative exists and is $\mathbb{P}$-almost everywhere uniquely given by
\begin{align}
    \frac{d\eta_A}{d\mathbb{P}}(\omega) = \nu_{\mathcal{C}}(\omega, A) \; . \label{eqn_regular_conditional_probability_as_rn_derivative}
\end{align}
Equivalently, with \cref{eqn_regular_conditional_probability_both_spaces} one can write the Radon-Nikod\'{y}m derivative on the mapped space $(E, \mathcal{E})$. With $S^{-1}(x) = \omega$, we arrive at
\begin{align} \label{eqn_rn_derivative_mapped_space}
    \frac{d\eta_A}{d\mathbb{P}}(S^{-1}(x)) = \left(\frac{d\eta_A}{d\mathbb{P}} \circ S^{-1} \right) (x) = \frac{d\eta^S_A}{d\mathbb{P}^S} (x) = \nu_{\mathcal{F}}(x, A) \; .
\end{align}
The second equality in \cref{eqn_rn_derivative_mapped_space} can easily be proven by using the change-of-variables formula.  Here the superscript $\cdot^S$ denotes the push forward to $S(\Omega)$ by $S$, i.e. $\eta_A^S = \eta_A \circ S^{-1}$. This notation will be kept throughout this work and also apply to multiple random variables in a straightforward way. 

This expression rigorously formalizes the notion of a conditional probability in a continuous setting, specifically conditioning to a specific point with respect to $\mathcal{C} / \mathcal{F}$. 

Further, since \cref{thrm_radon_nikodym} assures the existence of a suitable regular conditional probability independent of the sub-$\sigma$-algebra to condition on, one can drop the specifier $\mathcal{C}$, and work with only $\nu$ instead.

Additionally, we want to connect the regular conditional probabilities introduced here to probability densities as used in practical approaches. To that end, consider that all density functions consist of Radon-Nikod\'{y}m derivatives of a pushed forward probability measure with respect to an auxiliary measure $\lambda_S$. In the case of a random variable $S$ we obtain its density as $p(s) = \frac{d\mathbb{P}^S}{d\lambda_S}$. In many cases, this auxiliary measure is either Lebesgue's measure (continuous settings), or the Dirac measure. These cases correspond to $\Omega$ being uncountable or countable, respectively.

To make sure that a measure $\lambda_S$ with respect to which we take the Radon-Nikod\'{y}m derivative actually exists, we recall a famous isomorphism theorem \cite{standard_borel_spaces_isomorphism_theorem}:

\begin{theorem}[Isomorphism Theorem of Standard Borel spaces]
\label{thrm_isomorphisms_standard_borel_spaces}
Let $E, F$ be standard Borel spaces. Then there exists a Borel isomorphism between $E$ and $F$ if and only if $\text{card}(E) = \text{card}(F)$. \\
Specifically, any standard Borel space is either isomorphic to $(I, \mathcal{B}(I))$ for an interval $I \subset \mathbb{R}$, a countably infinite or finite standard Borel space $(\mathbb{S}, \mathcal{B}(\mathbb{S}))$ for a set $\mathbb{S} \subseteq \mathbb{N}$, or a disjoint union of both.
\end{theorem}

With \cref{thrm_isomorphisms_standard_borel_spaces} we can conclude that there must always exist a decomposition of the original probability measure into discrete and continuous parts. Moreover, we can construct a reference measure with respect to which we can sensibly define a density function $p$ as above.

\begin{corollary} \label{corollary_isomorphisms_standard_borel_spaces}

Let $\left((E, \tau_E), \mathcal{B}(E), \mathbb{P}^S \right)$ be a Radon Borel probability space. 

Then $\mathbb{P}^S$ has a decomposition into discrete and continuous parts, acting on the discrete and continuous components of $\left((E, \tau_E), \mathcal{B}(E) \right)$. More precisely, up to an isomorphism $\mathbb{P}^S$ can be taken to be of the form of a weighted combination of Dirac and Lebesgue measure, each supported precisely on the respective disjoint parts of the probability space which are either isomorphic to $(I, \mathcal{B}(I))$ or $(\mathbb{S}, \mathcal{B}(\mathbb{S})), \mathbb{S} \subseteq \mathbb{N}$ as in \cref{thrm_isomorphisms_standard_borel_spaces}. Call this isomorphism $\zeta$. Then 
\begin{align} \label{eqn_corollary_form_probability_measures}
    \mathbb{P}^S(A)= \sum\limits_{i=0}^{\text{card}(\mathbb{S})} p_{x_i}\delta^\zeta_{x_i}(\text{supp}_{\text{disc}} \cap A) + \int_{\text{supp}_{\text{cont}} \cap A} f(x) dm^\zeta(x) \;\; \forall A \in \mathcal{B}(E) \; .
\end{align} 
Here by $\delta^\zeta_x$ we denote the pushed forward Dirac measure only supported on $\{x\}$, by $m^\zeta$ the image Lebesgue measure, and by $\text{supp}_{\text{i}}$ the continuous and discrete parts of the probability space. Further we denote by $p_{x_i}$ the discrete probability masses of the atoms $\{x_i\}$ and by the non-negative function $f$ the continuous density. Note that $\sum_i p_{x_i} + \int_{E} f(x) dm^\zeta(x) = 1$. Moreover, considering the reference measure
\begin{align} \label{eqn_corollary_form_reference_measures}
    \lambda^S(A)= \sum\limits_{i=0}^{\text{card}(\mathbb{S})} \delta^\zeta_{x_i}(\text{supp}_{\text{disc}} \cap A) + m^\zeta(\text{supp}_{\text{cont}} \cap A) \;\; \forall A \in \mathcal{B}(E) \; ,
\end{align} 
we find that $\mathbb{P}^S \ll \lambda^S$.
Note that these supports are technically not necessary, however, as it appeals more to the intuition to separate the individual supports which make up $\left((E, \tau_E), \mathcal{B}(E) \right)$, we formulate the corollary in this way.
\end{corollary}

\begin{proof}
Then the proof is purely constructive; let again $\zeta$ be the isomorphism whose existence is guaranteed by \cref{thrm_isomorphisms_standard_borel_spaces}. If $\left((E, \tau_E), \mathcal{B}(E) \right)$ is isomorphic to an interval $\left(I, \mathcal{B}(I)\right)$ for $I \subset \mathbb{R}$, we take $\mathbb{P}^S$ to be a weighted Lebesgue measure on the image under $\zeta$ of subsets in $E$; if it is isomorphic to $(\mathbb{S}, \mathcal{B}(\mathbb{S}))$ for a set $\mathbb{S} \subseteq \mathbb{N}$, we take a weighted Dirac measure which has been pushed forward by $\zeta$. If it is isomorphic to a disjoint union $\left(I \cup \mathbb{S}, \mathcal{B}_{du}(I \cup \mathbb{S})\right)$ of both, where $ \mathcal{B}_{du}(I \cup \mathbb{S})$ denotes the Borel $\sigma$-algebra constructed via the disjoint union topology, one takes a linear combination of both acting on the preimages of $I$ and $\mathbb{S}$, respectively. Further, the measures $\mathbb{P}^S$ and $\lambda^S$ is readily $\sigma$-finite and hence the Radon-Nikod\'{y}m derivative of $\mathbb{P}^S$ with respect to $\lambda^S$ exists by \cref{thrm_radon_nikodym}.

\end{proof}

\section{A measure-theoretic generalization for $i^{\sx}_\cap$} \label{section_continuation_PID}

In this section, we extend PID to more general measure-theoretic settings. In particular, we derive a generalized variant of the measure of redundant information, $i_\cap^{\sx}$, introduced in  \cref{section_PID_isx_measure}, explain the setting investigated throughout this work (\cref{section_measure_theoretic_approach_subsection_setting}), and define $i_\cap^{\sx}$ for continuous random variables (\cref{section_measure_theoretic_approach_subsection_definition_sxpid}) and discrete-continuous mixtures. Then, we show that under consideration of finite, discrete spaces this leads back to the original discrete $i^{\sx}_\cap$ as in \cref{eqn_isx_discrete_case}. Furthermore, we prove that $i^{\sx}_\cap$ fulfills a pointwise version of the Williams and Beer axioms and several desirable properties such as invariance under isomorphisms and differentiability in a specific sense.

\subsection{The measure-theoretic setting} \label{section_measure_theoretic_approach_subsection_setting}
Since in this section we strive to expand the framework of PID as presented in \cref{section_PID_isx_measure} to a general measure-theoretic setting, we need to specify and describe the objects for which the theory developed here holds. These requirements themselves, however, are not particularly restrictive, as already standard settings such as the real numbers and finite sets can fulfill them with suitable assumptions.

Let there be a locally compact Radon space $(\Omega, \tau_\Omega)$ and a finite family $\{(E_{S_i}, \tau_{E_{S_i}})\}_{i \in \{1, \ldots, N_S\}} \cup (E_T, \tau_{E_T})$  of locally compact Radon spaces such that for $\Omega$ as well as each element of the family, every compact subset thereof is metrizable. That is, there exist metrics $d_\Omega$ and a family of metrics $\{d_{E_{S_i}}\}_{i \in \{1, \ldots, N_S\}} \cup \{d_{E_T}\}$ generating $\tau_\Omega$,  $\tau_{E_{S_i}} \forall i \in \{1, \ldots, N_S\}$ and $\tau_{E_T}$, respectively. Note that each individual Radon space may be uncountably large. 

Some of the most prominent examples for Radon spaces matching those conditions include $\mathbb{R}^n$ with the natural topology, or any finite or countably infinite set endowed with a metric. 

Associated to $(\Omega, \tau_\Omega)$ and $\{(E_{S_i}, \tau_{E_{S_i}})\}_{i \in \{1, \ldots, N_S\}} \cup \{(E_T, \tau_{E_T})\}$ we consider families of random variables, denoted by $\mathbf{S} = \{S_i\}_{i \in \{1, \ldots, N_S\}}$, called \emph{source random variables} and $T$, called \emph{target random variable} such that the image of each of the random variables is $X(\Omega) =: E_{X}, X \in \mathbf{S} \cup \{T\}$. Then we denote the full mapped space by $E := \bigtimes_{X \in \mathbf{S} \cup \{T\}} E_X$. Finally, denote by $E_{-Y}, Y \in \mathbf{S} \cup \{T\}$ the product of all $E_X$ without $Y$, i.e. $E_{-Y} = \bigtimes_{X \in \mathbf{S} \cup \{T\}\setminus \{Y\}} E_X$.

Note that for finitely many factors, the resulting topology $\tau_E$ on $E$ making it a Radon space is the product topology, and the Borel $\sigma$-algebra generated from this topology is the same as the product of the sigma algebras, i.e. $\mathcal{B}(E) = \bigtimes_i \mathcal{B}(E_{S_i}) \times \mathcal{B}(E_T)$. Then it is known that the product of these families is a Radon space too \cite{product_radon_spaces} and hence standard. 

Understanding $((\Omega, \tau_{\Omega}), \mathcal{B}(\Omega))$ together with a complete Borel probability measure $\mathbb{P}$ as a Radon Borel probability space then accounts for any sort of setting realistically imaginable, as the requirements above allow for purely discrete sample spaces, purely continuous sample spaces, as well as mixtures where either of the individual sample spaces may be discrete or continuous, or have itself discrete of continuous components. An example of the latter one being a random variable mapping a continuous space to $\mathbb{R}$, taking the value $0$ with probability $p$, and being uniformly distributed over $[1,2]$ with probability $1-p$. Hence the spaces we consider are precisely the objects from \cref{def_radon_borel_probability_spaces} and therefore the theorems from \cref{section_mathematical_background} all hold for these spaces.

\subsection{Definition of $i^{\sx}_\cap$ for continuous random variables} \label{section_measure_theoretic_approach_subsection_definition_sxpid}
To write down the shared exclusions' redundant information from \cref{eqn_isx_discrete_case}, we search for analogies with the discrete setting. In order to define a possible candidate of generalization we need a regular conditional probability as well as the aforementioned auxiliary random variable to exist, such that we can sensibly state the expressions. Thus we define the auxiliary random variable that has been mentioned in \cref{eqn_isx_discrete_case_auxiliary_random_variable} formally. After that, we use an analogy with the full local mutual information to compose the quantities treated so far in a manner that resembles \cref{eqn_isx_discrete_case_auxiliary_random_variable}. Finally, we arrive at a generalization of the local shared exclusions' redundant information.

In order to have the expressions stated in a way that includes the statistical interplay between the source and target random variables, as formalized by \cref{eqn_isx_discrete_case_auxiliary_random_variable}, we need to write down the corresponding expression in terms of the events on $\Omega$. 

\begin{definition}[Injection] \label{def_injections}
Consider the setting as in \cref{section_measure_theoretic_approach_subsection_setting}, including the family of random variables $\mathbf{S}$, $T$ and the corresponding Radon Borel probability spaces $\left( (\Omega, \tau_\Omega), \mathcal{B}(\Omega)\right)$ and $\left( (E, \tau_E), \mathcal{B}(E)\right)$. Let $X \in \mathbf{S} \cup \{T\}$ and recall that $E$ is the product of all the  $E_X$. Then we introduce the shorthand notation for the preimage under $X$ as
\begin{align*}
    \Pi_{X}: \mathcal{B}(E_X) \to \mathcal{B}(\Omega), \; B \mapsto \Pi_{X}(B) = \left\{\omega \in \Omega : X(\omega) \in B \right\} \; .
\end{align*}

\end{definition}
As the logical structure of the discrete redundant information in \cite{makkeh2021isx} is complex in the sense of utilizing logical disjunctions, the technique of addition for different constraints connected via a logical OR cannot generally be applied in a straightforward manner for regular conditional probability measures as introduced in \cref{thrm_regular_conditional_probability}. This technique would amount to evaluation of a function at multiple evaluation points in the same instance. To avoid this, we introduce an auxiliary random variable:

\begin{definition}[Auxiliary Indicator Variable] \label{def_auxiliary_random_variable}
Suppose we have a fixed $s \in E_{-T}$, then $s$ has the form $s = (s_k)_{k \in \left\{1, \ldots, N_S\right\}}$. 

Denote by $R_s$ the event 
\begin{align*}
    R_s := \bigcup\limits_{i=1}^{N_S} \Pi_{S_i} \left(\{s_i\}\right) \subset \mathcal{B}(\Omega) \; .
\end{align*}

Similarly we define $R_{\alpha, s} $ for a collection of indices of random variables $\alpha = \{\mathbf{a}_1, \ldots, \mathbf{a}_m\} \in \mathcal{P}(\mathcal{P}(\{1, \ldots, N_S\}))$ where $\mathbf{a}_k$ denotes a set of indices by 
\begin{align*}
R_{\alpha,s} := \bigcup\limits_{k\in[m]} \bigcap\limits_{j \in \mathbf{a}_k} \Pi_{S_j}(\{s_j\}) \; .
\end{align*}

Then we introduce the short notation $\mathds{1}_s = \mathds{1}_{R_s}$ and $\mathds{1}_{\alpha, s} = \mathds{1}_{R_{\alpha, s}}$ for the indicator random variables with respect to $R_s$ and $R_{\alpha, s}$.

\end{definition}
 
Due to the restrictions from the setting present, there exists a regular conditional probability with respect to a disintegration from $\Omega$ to $E$ via application of $\mathds{1}_s$ as in \cref{thrm_disintegration}. 

\begin{definition}[Marginal Measures]
Let again $X \in \mathbf{S} \cup \{T\}$ and fix an $s \in E$. Note that with the vector of random variables $V := (S_1, \ldots, S_{N_S}, T)$ and \cref{def_injections} we can write the marginal measure with respect to $X$ as 
\begin{align*}
    \mathbb{P}^X : \mathcal{B}\left(E_X\right) \to [0,1], \; B \mapsto \mathbb{P}^X(B) = \mathbb{P} \circ \Pi_X(B) \; .
\end{align*}

Additionally, we call $\nu_s$ the regular conditional probability measure on $\mathcal{B}\left(\Omega\right)$, conditioned on $R_s$. From this, we obtain the image measure $\nu_s^V = \nu_s \circ V^{-1}$ on $\mathcal{B}(E)$.
Further, define $\nu^X_{s}(B) := \nu_{s} \circ \Pi^{-1}_X(B), B \in \mathcal{B}(E_X)$ to be the marginal measure on $E_X$ with respect to $X$.

In the exact same way we define the quantities $\nu_{\alpha, s}, \nu^V_{\alpha, s}$ and $\nu^X_{\alpha, s}$ as regular conditional probabilities conditioned on $R_{\alpha, s}$ (or rather $\mathds{1}_{\alpha, s} = 1$).

Note that formally, by conditioning onto the $\sigma$-algebra generated by $R_s$, for an $A \in \mathcal{B}(\Omega)$ we would generate a random measure $\nu_s(\omega, A)$ as in \cref{eqn_regular_conditional_probability_as_rn_derivative}. Since we condition on $R_{s}$, however, this measure is only non-zero when $\omega \in R_s$. So technically we should be writing $\frac{d\eta^{\mathds{1}_s}_A}{d\mathbb{P}^{\mathds{1}_s}}(1) = \nu_s(1, A)$. Nonetheless, to improve the readability, we omit the first argument and proceed only writing $\nu_s(A)$, implicitly carrying the 1. We do the same for all regular conditional probability measures defined here.
\end{definition}

Then $\nu_{s}$ enables us to measure the sets in $\Omega$  under the condition, that either of the realizations $\{s_i\} =: s \in E_{-T}$ are taken under the random variables in the family $\textbf{S}$.

Further note that $\nu_{\alpha, s} \ll \mathbb{P}$ for all $\alpha \in \mathcal{P}(\mathcal{P}(\{1, \ldots, N_S\})), s \in E$. This can be seen when recalling the form of $\nu_{\alpha, s}$ from \cref{eqn_regular_conditional_probability_as_rn_derivative}. Explicitly, let $N \in \mathcal{B}(\Omega)$ be a $\mathbb{P}$-nullset. Then it follows that $\eta_N(A) = \mathbb{P}(N \cap A) = 0$ for all $A \in \mathcal{B}(\Omega)$ and consequently $ \nu_{\alpha, s}(N) = \frac{d\eta^{\mathds{1}_{\alpha, s}}_N}{d\mathbb{P}^{\mathds{1}_{\alpha, s}}}(1) = 0$. Note that here, as explained above, we again omitted the argument 1.
Thus all $\mathbb{P}$-nullsets are also $\nu_{\alpha, s}$-nullsets, with the same readily holding for the image and marginal measures, i.e. $\nu_{\alpha, s}^V \ll \mathbb{P}^V$ and $\nu_{\alpha, s}^X \ll \mathbb{P}^X$ for all $X \in \mathbf{S} \cup \{T\}$. 

In addition, recall from \cref{corollary_isomorphisms_standard_borel_spaces} that, as probability densities are defined to be Radon-Nikod\'{y}m derivative of probability measures with respect to some reference measure, we can generate a suitable measure $\lambda_T$ with respect to which $\mathbb{P}^T$ and the $\nu^T_{s}$ are absolutely continuous. When $\mathbb{P}^T$ has the form from \cref{eqn_corollary_form_probability_measures}, the corresponding $\lambda^T$ has a form as in \cref{eqn_corollary_form_reference_measures}.

In order to define a candidate for an information measure of the redundant part of a local mutual information, we will use the notation introduced above to find the corresponding form of the full local mutual information. Then we will, based on the similarities between the rewritten form of the local mutual information and the discrete local mutual information, propose a new measure-theoretically well-defined local redundant information measure.

Recall that for any $s_i\in E_{S_i},$ given that $S_i$ is discrete, the local mutual information takes the form 
\begin{align} \label{eqn_local_mutual_information_conditional_density}
    i(t:s_i) = \log \left[ \frac{p(t\mid s_i)}{p(t)} \right] \; .
\end{align}
This $i(t:s_i)$ can be generalized for any $S_i$ (not necessarily discrete) as follows. Simplify, for the moment, the setting build up in \cref{section_measure_theoretic_approach_subsection_setting} back to the case of a single source random variable called $S_i$, and the target random variable $T$. Then the regular conditional probability measure on $E_T$ conditioned on $R_{s_i}$, denoted by $\nu^T_{s_i}$, formalizes the measure from which the conditional density $p(t|s_i)$ in \cref{eqn_local_mutual_information_conditional_density} is derived. Then we write 
\begin{align}
    i(t:s_i) &= \log \left[ \frac{ \frac{d\nu^T_{s_i}}{d\lambda_T}(t) }{  \frac{d\mathbb{P}^T}{ d\lambda_T} (t) } \right] \label{eqn_local_MI_differential_form} \\
    &= \log \left[ \frac{d\nu^T_{s_i}}{d\mathbb{P}^T} (t) \right] \; . \label{eqn_local_mutual_information_rewritten_realizations}
\end{align}
Similarly, back in the full setting from \cref{section_measure_theoretic_approach_subsection_setting}, any realization $(\mathbf{s}, t)\in \bigtimes_{j\in\mathbf{a}}E_{S_j}\times E_{T},$ and the corresponding regular conditional probability $\nu^T_{\alpha, s}$, where $\alpha = \{\mathbf{a}\}$, give rise to the local mutual information the realization of the joint random variable $\{S_j\}_{j \in \mathbf{a}}$ has about the realization of the target variables $T$ as
\begin{align}
    i(t:\mathbf{a}) = \log \left[ \frac{d\nu^T_{\alpha, s}}{d\mathbb{P}^T}\left(t\right) \right] \; . \label{eqn_local_mutual_information_rewritten_collections}
\end{align}
From this form for a single collection of indices we can directly derive an expression for a number of mutually inclusive disjunctions of conditions of which either can be true, as in \cref{eqn_isx_discrete_case_auxiliary_random_variable}.

Then, to generate a measure for the local redundant information, in direct analogy with \cref{eqn_local_mutual_information_rewritten_realizations}, we define: 

\begin{definition} \label{def_continuous_PID}
Let the notation be as above. Based on the idea of the discrete measure of shared exclusions of realizations
\begin{equation*}
    i^{\sx}_{\cap}(t : \{1\};\ldots;\{n\}) := \log \left[ \frac{p\left( \{T = t\}\mid\bigvee_i\{S_i = s_i\} \right)}{p \left( \{T = t\} \right)} \right] \; ,
\end{equation*}
where $\bigvee_i$ denotes a logical or, we define a measure-theoretic information-measure for redundant information, 

\begin{align}
    i^{\sx}_\cap(t : \{1\};\ldots;\{n\}) = \log \left[ \frac{d\nu^T_{s}}{d\mathbb{P}^T} \left( t \right) \right] \; . \label{eqn_continuous_PID_realizations}
\end{align}

Analougusly we define the shared exclusions' redundant information under consideration of collections $\alpha = \{\mathbf{a}_1,\ldots,\mathbf{a}_m\},$ using $\nu^T_{\alpha, s}$ as the image probability measure, generated from disintegration of $\Omega$ with respect to $\mathds{1}_{\alpha, s} = 1$ from \cref{def_auxiliary_random_variable}. This generalizes the local redundant information \cref{eqn_continuous_PID_realizations} as follows:
\begin{align}
    i^{\sx}_\cap(t: \alpha) = i^{\sx}_\cap(t : \mathbf{a}_1;\ldots;\mathbf{a}_m) := \log \left[ \frac{d\nu^T_{\alpha, s}}{d\mathbb{P}^T} \left( t \right) \right] \; . \label{eqn_continuous_PID_collections}
\end{align}
Moreover, in the case that $\nu^T_{\alpha, s}$ is not absolutely continuous with respect to $\mathbb{P}^T$, we define that $i^{\sx}_\cap(t: \alpha) = \infty$, and for those $t \in E_T$ for which $\frac{d\nu^T_{\alpha, s}}{d\mathbb{P}^T} \left( t \right) = 0$, we define $i^{\sx}_\cap(t: \alpha) = -\infty$.  

Analogous to $\frac{p(t|s)}{p(t)} = \frac{p(s|t)}{p(s)}$ by Bayes' law, an equivalent form of $i^{\sx}_\cap(t:\alpha)$ is

\begin{align}
    i^{\sx}_\cap(t:\alpha) = \log \left[ \frac{d\nu^T_{\alpha, s}}{d\mathbb{P}^T}(t) \right] = 
    \log \left[ \frac{d\nu^{\mathds{1}_{\alpha, s}}_{t}}{d\mathbb{P}^{\mathds{1}_{\alpha, s}}}(1) \right] \; . \label{eqn_monotonicity_bayes_reformulation}
\end{align}
In the latter expression we use the following: 

As $\frac{d\nu^T_{\alpha, s}}{d\lambda_T}(t) = \frac{d\nu^{\mathds{1}_{\alpha, s}}_{t}}{d\mathbb{P}^{\mathds{1}_{\alpha, s}}}(1) \frac{d\mathbb{P}^T}{d\lambda_T}(t)$, we can divide both expressions by $\frac{d\mathbb{P}^T}{d\lambda_T}(t)$ to obtain $\frac{d\nu^T_{\alpha, s}}{d\mathbb{P}^T}(t) = \frac{d\nu^{\mathds{1}_{\alpha, s}}_{t}}{d\mathbb{P}^{\mathds{1}_{\alpha, s}}}(1)$.
Here $\nu^{\mathds{1}_{\alpha, s}}_{t}$ is a disintegrated probability measure concentrated, hence conditioned, on $t$, then pushed forward by $\mathds{1}_{\alpha, s}$. This idea is similar to a remark above theorem 2 in \cite{disintegration_theorem_chang_pollard}. The measure $\mathbb{P}^{\mathds{1}_{\alpha, s}}$ is simply the original measure, pushed forward by $\mathds{1}_{\alpha, s}$. 
\end{definition}

\subsection{Necessary and desired properties} \label{section_continuation_PID_subsection_properties}

From the form that has been derived above, one can now conclude a number of properties, which also distinguish the original measure from the discrete case. Some of the properties mentioned are necessary for each conceptually sensible PID information-measure, and some are simply beneficial for applications.

\subsubsection{Recovering the discrete definition} \label{subsubsection_recovering_discrete_case}
A preliminary sanity check that the decomposition presented actually represents a generalization of the discrete original approach, we verify that the continuous shared exclusions' redundant information simplifies to the discrete expression if the probability measure $\mathbb{P}$ is a discrete measure, with the underlying space $\Omega$ being of finite cardinality.

Explicitly, take $\mathbb{P} = \sum_{\omega \in \Omega} p_\omega \delta_{\{\omega\}}$ where $\delta_{\{\omega\}}(A) := \begin{cases} 1 & \text{ if }\omega \in A \\ 0 & \text{ otherwise} \end{cases}$ is the Dirac measure and $\sum_{\omega \in \Omega} p_\omega= 1$. 
Understanding $\Omega$ as a topological space together with the discrete topology, just as done with $\{0,1\} = \mathrm{im}(\mathds{1}_{\alpha, s})$ above, leads to the existence of a regular conditional probability on the corresponding Borel $\sigma$-algebra, conditioned on $R_{\alpha, s}$. Call this regular conditional probability measure $\mathbb{D}$ when conditioned on $\mathds{1}_{\alpha, s}=1$. The partial density stemming from this regular conditional probability, i.e. $\mathbb{D}^T := \mathbb{D}^V \circ \Pi_T$ is then a measure which measures sets of $E_T$, conditioned on $R_{\alpha, s}$. The discrete nature of this measure then lets us conclude that $\mathbb{D}^T$ evaluates to 
\begin{align*}
    \mathbb{D}^T (A) = \sum_{\omega \in R_{\alpha, s} \cap T^{-1}(A)} p_\omega \; .
\end{align*}
From this measure, the density can be calculated as Radon-Nikod\'{y}m derivative $\frac{d\mathbb{D}^T}{d\lambda_T}(t)$.
This latter expression exactly equals, in the more common notation of \cite{makkeh2021isx}, $p\left(t \mid \mathrm{im}(\mathds{1}_{\alpha, s}) = \{1\}\right)$.

The usual probability density, formed as $p(t) = \frac{d\mathbb{P}^T}{d\lambda_T}(t)$ remains exactly the same, even notationally.

Thus we arrive at the very same expression for the discrete case, as was presented in \cite{makkeh2021isx}, verifying the generalizing nature of the measure-theoretic redundant information measure from \cref{def_continuous_PID}.

\subsubsection{Self-redundancy, monotonicity, and symmetry under permutations} \label{subsubsection_self_redundancy_monotonicity_symmetry}
Recall from \cref{sec:pid-redundant-info} that any redundant measure should fulfill three properties, namely, self-redundancy, invariance under superset removal or addition, and symmetry  so that such a redundancy measure respects the parthood nature of the PID quantities. We show here that $i^{\sx}_\cap$ fulfils these considerations for any point in $\Omega$ which are then trivially fulfilled for the average $I^{\sx}_\cap.$ We first state the pointwise version of these considerations. 

\begin{definition} \label{def_ppid_axioms}
Consider the shared exclusions measure for redundant information defined in \cref{def_continuous_PID}. Then, for any $(s,t)\in\Omega$, $i^{\sx}_\cap$ satisfies:
\begin{itemize}
    \item \textbf{Self-redundancy} if $i^{\sx}_\cap(t:\mathbf{a}) = i(t:\mathbf{a})$ \; ,
    \item \textbf{invariance under superset removal or addition} if, whenever $\exists j \in [m]$ such that $\mathbf{a}_{j} \subseteq \mathbf{a}_{m+1}$, then $i^{\sx}_\cap(t:\mathbf{a}_1;\ldots;\mathbf{a}_m)= i^{\sx}_\cap(t:\mathbf{a}_1;\ldots;\mathbf{a}_m;\mathbf{a}_{m+1}),$ 
    \item and \textbf{symmetry} if $i^{\sx}_\cap(t:\mathbf{a}_1;\ldots;\mathbf{a}_m) = i^{\sx}_\cap(t:\mathbf{a}_{\kappa(1)};\ldots;\mathbf{a}_{\kappa(m)})$ for all permutations $\kappa$.
\end{itemize}

\end{definition}
In the following theorem we show that $i^{\sx}_\cap(t:\alpha)$ satisfies these considerations. 
\begin{theorem} \label{thrm_local_pid_axioms}
    $i^{\sx}_\cap(t:\alpha)$ satisfies \textbf{self-Redundancy}, \textbf{invariance under superset removal or addition}, and \textbf{symmetry}.
\end{theorem}
\begin{proof}
To show that $i^{\sx}_\cap$ satisfies \textit{self-redundancy}, comparing \cref{eqn_local_mutual_information_rewritten_realizations} and \cref{eqn_continuous_PID_realizations}, or \cref{eqn_local_mutual_information_rewritten_collections} and \cref{eqn_continuous_PID_collections}, note
\begin{align*}
    i^{\sx}_\cap(t:s_i) &= \log \left[ \frac{d\nu^T_{s_i}}{d\mathbb{P}^{T}} \left( t \right) \right] = i(t:s_i) \quad \text{as well as } \\
    i^{\sx}_\cap(t:\alpha) &= \log \left[ \frac{d\nu^T_{\alpha, s}}{d\mathbb{P}^{T}} \left(t\right) \right] = i(t:(s_i)_{i\in\mathbf{a}}) \; ,
\end{align*}
where $\alpha = \{\mathbf{a}\}$. Therefore, the redundant shared information for single realizations or collections simplifies to the local mutual information.

To show that $i^{\sx}_\cap$ satisfies \textit{invariance under superset removal or addition}, consider 

$R_{\alpha \cup \textbf{a}_{m+1}, s} = R_{\alpha, s}$ since $\textbf{a}_{m+1} \supset \textbf{a}_j$ for some $\textbf{a}_{j}\in\alpha.$ Thus
\begin{align*}
    i^{\sx}_\cap(t:\alpha; \textbf{a}_{m+1})    &=\log \left[ \frac{d\nu^{\mathds{1}_{\alpha \cup \textbf{a}_{m+1}, s}}_{t}}{d\mathbb{P}^{\mathds{1}_{\alpha \cup \textbf{a}_{m+1}, s}}}(1) \right] = \log \left[ \frac{d\nu^{\mathds{1}_{\alpha, s}}_{t}}{d\mathbb{P}^{\mathds{1}_{\alpha, s}}}(1)\right] = i^{\sx}_\cap(t:\alpha) \; .
\end{align*}
Finally, to show that $i^{\sx}_\cap$ satisfies \textit{symmetry}, consider that the independence of the order of collections can be readily seen, as unions and intersections are individually commutative. Therefore, the random variable $\mathds{1}_{\kappa(\alpha), s}$ equals $\mathds{1}_{\alpha, s}$ pointwise\footnote{$\kappa(\alpha) := \{\mathbf{a}_{\kappa(1)}, \ldots, \mathbf{a}_{\kappa(m)}\}$}. Equivalently, $R_{\kappa(\alpha), s} = R_{\alpha, s}$. Thus one can infer that there must hold $i^{\sx}_\cap(t:\kappa(\alpha)) =  i^{\sx}_\cap(t:\alpha)$ for all permutations $\kappa$.
\end{proof}

An additional property was proposed by Williams and Beer as a requirement for a redundancy measure. This property is called \textit{monotonicity} and reads as follows: a redundancy measure $I_\cap$ is monotonically decreasing under the addition of $\mathbf{a}_{m+1}$ that is not a superset of any $\mathbf{a}_j\in\alpha.$ 

The conceptual reasoning behind monotonicity is that when the redundant information of the collections in $\alpha$ is nonnegative then adding any additional $\mathbf{a}_{m+1}$ can only shrink this redundancy. However, pointwise measures such as $i^{\sx}_\cap$ can take negative values since locally sources can misinform about the target. For the interpretability of such negativity the pointwise measures should be the difference of two nonnegative parts: the informative and misinformative part. Then, when the measure takes a negative value, it means that its misinformative part over took its informative part and vice-versa. Therefore the monotonicity reasoning conceptually fails on the pointwise measure itself since nothing forbids the additional $\mathbf{a}_{m+1}$ to yield a decrement in the misinformative that is less than the decrement of the informative part and thus leads to an increase in the redundancy.  Nevertheless, monotonicity should still hold for each of the informative and misinformative parts.

To this end, consider that\footnote{The notation $\hat{\cdot}$ is introduced due to the following reasons. Applying the self-redundancy property $\hat{i}^{\sx\pm}_{\cap}$ reduces to differential entropy that (i) lacks the invariant property and (ii)can take on negative values. To remedy this, we need to introduce $i^{\sx\pm}_{\cap} = \hat{i}^{\sx\pm}_{\cap} + F$ where $i^{\sx\pm}_\cap$ are both non-negative functions and $F$ is chosen appropriately.} $i^{\sx}_\cap = \hat{i}^{\sx+}_{\cap} - \hat{i}^{\sx-}_{\cap}$ by decomposing the Radon-Nikod\'{y}m derivative \cref{eqn_monotonicity_bayes_reformulation}, with respect to a reference measure $\lambda$ via

\begin{align}
    i^{\sx}_\cap(t:\alpha) &= \log \left[ \frac{d\nu^{\mathds{1}_{\alpha, s}}_{t}}{d\mathbb{P}^{\mathds{1}_{\alpha, s}}}(1) \right] = \log \left[ \frac{\frac{d\nu^{\mathds{1}_{\alpha, s}}_{t}}{d\lambda^{\mathds{1}_{\alpha, s}}}(1)}{\frac{d\mathbb{P}^{\mathds{1}_{\alpha, s}}}{d\lambda^{\mathds{1}_{\alpha, s}}}(1)} \right] = - \underbrace{ \log \left[\frac{1}{\frac{d\nu^{\mathds{1}_{\alpha, s}}_{t}}{d\lambda^{\mathds{1}_{\alpha, s}}}(1)} \right]}_{\text{$:= \hat{i}^{\sx-}_{\cap}$}} + \underbrace{\log \left[ \frac{1}{\frac{d\mathbb{P}^{\mathds{1}_{\alpha, s}}}{d\lambda^{\mathds{1}_{\alpha, s}}}(1)} \right]}_{\text{$:= \hat{i}^{\sx+}_{\cap}$}} \; . \label{eqn_isxpid_decomposition}
\end{align}

Note that this decomposition exists precisely when $ \frac{d\mathbb{P}^{\mathds{1}_{\alpha, s}}}{d\lambda^{\mathds{1}_{\alpha, s}}}(1) \neq 0 \neq \frac{d\nu^{\mathds{1}_{\alpha, s}}_{t}}{d\lambda^{\mathds{1}_{\alpha, s}}}(1)$. This, however, was exactly the condition for $i^{\sx}_\cap$ as in \cref{eqn_monotonicity_bayes_reformulation} to exist in the first place.

Therefore, we prove the monotonicity for both these parts individually:
\begin{proposition} \label{proposition_monotonicity}
Recall the decomposition as in \cref{eqn_isxpid_decomposition}. Then both the $\hat{i}^{\sx\pm}_{\cap}$ are monotone under extension of the collections considered as in \cref{def_ppid_axioms}.
\end{proposition}

\begin{proof}

First recall that the monotonicity from \cref{def_ppid_axioms} is equivalent to the statement that
\begin{align}
    \hat{i}^{\sx\pm}_\cap(t:\alpha) \leq \hat{i}^{\sx\pm}_\cap(t:\alpha; \textbf{a}_{m+1}) \label{eqn_monotonicity_condition_isxpm}
\end{align}
for $\textbf{a}_{m+1}$ being another collection under consideration. 

If $\textbf{a}_{m+1}$ is a superset of any of the $\textbf{a}_{j} \in \alpha$, then there is nothing to show due to \cref{thrm_local_pid_axioms}. If there is no $\textbf{a}_j \in \alpha$ such that $\textbf{a}_j \subset \textbf{a}_{m+1}$, then we proceed as follows.

First treating $\hat{i}^{\sx+}_\cap$, \cref{eqn_monotonicity_condition_isxpm} is equivalent to 
\begin{align*}
    0 \geq \hat{i}^{\sx+}_\cap(t:\alpha) - \hat{i}^{\sx+}_\cap(t:\alpha; \textbf{a}_{m+1}) = \log \left[ \frac{d\nu^{\mathds{1}_{\alpha, s}}_{t}}{d\nu^{\mathds{1}_{\alpha \cup \textbf{a}_{m+1}, s}}_{t}}(1) \right] \; .
\end{align*}
The latter Radon-Nikod\'{y}m derivative exists, since $R_{\alpha, s} \subseteq R_{\alpha \cup \textbf{a}_{m+1}, s}$. Hence the random variables support is a larger one for the measure found in the denominator.
We need to $ 1 \geq \frac{d\nu^{\mathds{1}_{\alpha, s}}_{t}}{d\nu^{\mathds{1}_{\alpha \cup \textbf{a}_{m+1}, s}}_{t}}(1)$ to hold, such that the negative logarithm is above zero.
Due to $\sigma$-additivity of $\nu_t$ and the inclusion of $R_{\alpha, s}$ in $R_{\alpha \cup \textbf{a}_{m+1}, s}$, this necessary leads to $ 1 \geq \frac{d\nu^{\mathds{1}_{\alpha, s}}_{t}}{d\nu^{\mathds{1}_{\alpha \cup \textbf{a}_{m+1}, s}}_{t}}(1)$ as claimed. Hence $i^{\sx+}_\cap(t:\alpha)$ is monotone under consideration of additional collections. Regarding $\hat{i}^{\sx-}_\cap(t:\alpha)$, the exact same argumentation holds, replacing $\nu^{\mathds{1}_{\alpha, s}}_{t}$ by $\mathbb{P}^{\mathds{1}_{\alpha, s}}$. Hence also $\hat{i}^{\sx-}_\cap(t:\alpha)$ is monotone according to \cref{def_ppid_axioms}.

\end{proof}

\subsubsection{Invariance under invertible measurable maps} \label{subsection_properties_subsubsection_invariance}

Invariance under the mentioned maps is highly beneficial property, specifically for applications on real-world data, since this property implies independence of the unit one measures in, and translation invariance of the data. As long as the relative frequencies of the data points is not meddled with, one can bijectively map them anywhere, generate a suitable measure, and apply the theory developed here. 

\begin{proposition} \label{proposition_invariance_local_sxpid}
The shared exclusions based redundant information measure from \cref{def_continuous_PID} is invariant under invertible measurable maps.
\end{proposition}

\begin{proof}

The property of invariance under this sort of maps with respect to the random variables $\textbf{S}, T$ and realizations $s, t$ can then readily been observed when realizing, that \cref{eqn_continuous_PID_realizations} and \cref{eqn_continuous_PID_collections} do not involve the auxiliary measure $\lambda_\mathbf{S}$ or $\lambda_T$.

Another way of seeing this, is to recall that all the $\sigma$-algebras we condition on, remain the same since they lie on $\Omega$ an are not affected by transformations carried out on $E$.

For a more sophisticated proof, see the following.

First note that by Souslin's studies of analytic sets\cite{souslin_analytic_sets}, any invertible, hence bijective measurable map is a topological isomorphism, and hence bicontinuous in the setting present. Due to this, the transformed space has a natural topology and thus Borel $\sigma$-algebra given via push forwards, making the transformed space a Radon Borel measurable space, and hence the theory developed in \cref{section_mathematical_background} is also valid for the transformed spaces. 

Then consider families of invertible, measurable, maps $\phi = \{\phi_i\}$ for the family $\{E_{S_i}\}$ and $\psi$ for $E_T$. Then introduce the notation $\tilde{s} = \{\phi_i(s_i)\}$ and $\tilde{t} = \psi(t_j)$ for the realizations. Thus, purely by notation, the transformed \cref{eqn_continuous_PID_realizations} becomes 
\begin{align*}
    i^{\sx}_{\cap}(\tilde{t} : \tilde{s}) = \log \left[ \frac{d\nu^{\psi \circ T}_{\tilde{s}}}{d\mathbb{P}^{\psi \circ T}} \left(\,\tilde{t}\,\right) \right] \; .
\end{align*}

By \cref{def_injections} it is evident that, since the marginal measures utilized, in the push forward, contain the preimage under $T$, the transformed mapping $\Pi_{\psi \circ T}\left(\psi(B)\right)$ contains the preimage of $\psi \circ T$ of $\psi(B)$, that is $(\psi \circ T)^{-1}(\psi(B)) = T^{-1}(B)$ for $B \in \mathcal{B}(E_T)$, with the same holding for $\mathbb{P}^T$. Hence, $\mathbb{P}^{T}$ - almost everywhere there holds  

\begin{align*}
    \frac{d\nu^{\psi \circ T}_{\tilde{s}}}{d\mathbb{P}^{\psi \circ T}} \left(\,\tilde{t}\,\right) = \frac{d\nu^T_{\tilde{s}}}{d\mathbb{P}^{T}} \left( t \right)
\end{align*}

on $\mathcal{B}(E_T)$ due to the uniqueness in \cref{thrm_disintegration}.

Further, the changes in $\tilde{s}$ are covered by a similar argument. As $\nu^T_{s}$ was constructed via $R_s = \bigcup\limits_{i=1}^{N_S} \Pi_{S_i}\left(\{s\}\right)$, any change by an invertible map $\phi_i$ is immediately undone in the set description of $R_s$. Hence, under a transformation $\{S_i\} \to \{\phi_i(S_i)\}$, we have $R_{\tilde{s}} = R_s$, and thus $\nu^V_{\tilde{s}, T} = \nu^V_{s, T}$. 

With the arguments above, it becomes clear that
\begin{align*}
    i^{\sx}_{\cap}(\tilde{t} : \tilde{s}) = \log \left[ \frac{d\nu^{\psi \circ T}_{\tilde{s}}}{d\mathbb{P}^{\psi \circ T}} \left( \, \tilde{t} \, \right) \right]  =  \log \left[ \frac{d\nu^T_{s}}{d\mathbb{P}^{T}} \left( t\right)  \right] = i^{\sx}_\cap(t:s) \; .
\end{align*}
The same argumentation holds exactly for collections $\alpha=\{\mathbf{a}_1,\ldots,\mathbf{a}_m\}$, with $R_{\tilde{s}} = R_s$ replaced by 
$R_{\alpha, \tilde{s}} = R_{\alpha, s}$ resulting in $\nu^T_{\alpha, \tilde{s}} = \nu^T_{\alpha, s}$ on $\mathcal{B}(E_T)$. Hence, there also holds $i^{\sx}_{\cap}(\tilde{t} : \alpha) =  i^{\sx}_{\cap}(t:\alpha)$.

\end{proof}

\subsubsection{Differentiability of $i^{\sx}_\cap$ with respect to probability densities} \label{subsubsection_differentiability_of_isx}
For this measure of pointwise redundant information to be feasibly applied in neural networks, it is crucial that $i^{\sx}_\cap$ is differentiable with respect to the underlying probability density, or equivalently, with respect to the probability measure $\mathbb{P}$. Practically, this would be the density function one would try to approximate via data. 
By differentiability, we here mean that a slight shift in the underlying probability measure $\mathbb{P}$ in the direction of another probability measure $\mathbb{Q}$ of strength $\epsilon$, i.e. $\mathbb{P} \to \mathbb{P} + \epsilon \mathbb{Q}$, amounts for a finite change of $i^{\sx}_\cap(t:\alpha)[\mathbb{P}]$ as a functional of $\mathbb{P}$. Explicitly, we will consider the derivative in its form $\lim\limits_{\epsilon \to 0} \frac{i^{\sx}_\cap(t:\alpha)[\mathbb{P} + \epsilon \mathbb{Q}] - i^{\sx}_\cap(t:\alpha)[\mathbb{P}]}{\epsilon}$. This form is equivalent to $\left. \big(\partial_\epsilon  i^{\sx}_\cap(t:\alpha)[\mathbb{P} + \epsilon \mathbb{Q}] \big) \right\vert_{\epsilon=0}$.

\begin{proposition} \label{proposition_differentiability_local_sxpid}
 The shared exclusions based redundant information measure from \cref{def_continuous_PID} is differentiable with respect to the underlying probability measure $\mathbb{P}$ and thus, the resulting probabilty density $\frac{d\mathbb{P}^V}{d\lambda}$ for a suitable $\lambda$. 
\end{proposition}

\begin{proof}
Here we will first consider differentiability with respect to the underlying probability measure $\mathbb{P}$, and then extend the notion to densities in a straightforward way.

First note that the set of all $\sigma$-finite signed measures on a measurable space forms a vector space, in which the measures which are absolutely continuous with respect to a reference measure $\gamma$, form a linear subspace. This subspace we will denote by $M_\gamma$. Further note that $M_\mu \subset M_\gamma \iff \mu \ll \gamma$, i.e. one obtains a decreasing family of linear subspaces in correspondence to absolute continuity. Let $\mathbb{Q} \in M_{\mathbb{P}}$ be a probability measure. Then we will consider the derivative $\lim\limits_{\epsilon \to 0} \frac{i^{\sx}_\cap(t:\alpha)[\mathbb{P} + \epsilon \mathbb{Q}] - i^{\sx}_\cap(t:\alpha)[\mathbb{P}]}{\epsilon}$ to calculate the equivalent of a directional derivative with respect to $\mathbb{Q}$, denoted by $\partial^{\mathbb{Q}}_{\mathbb{P}}$.

From \cref{eqn_continuous_PID_collections} and \cref{eqn_local_MI_differential_form} we see that evaluating $\partial^{\mathbb{Q}}_{\mathbb{P}} i^{\sx}_\cap(t:\alpha)[\mathbb{P}]
$ is the same as finding the above limit for $\log \left[\frac{d\mathbb{P}^T}{d\lambda_T} \right]$ and $\log \left[ \frac{d\nu_{\alpha, s}^T}{d\lambda_T} \right]$ individually. Defining $\mathbb{Q}^T = \mathbb{Q} \circ \Pi_T$, starting with the marginal density $p_T(t) = \frac{d\mathbb{P}^T}{d\lambda_T}(t)$ and analogously $q_T(t) = \frac{d\mathbb{Q}^T}{d\lambda_T}(t)$, we get
\begin{align}
    \partial^{\mathbb{Q}}_{\mathbb{P}} \log \left[ \frac{d\mathbb{P}^T}{d\lambda_T} (t) \right] &= \lim\limits_{\epsilon \to 0} \frac{1}{\epsilon} \left( \log \left[ \frac{d\left(\mathbb{P}^T + \epsilon \mathbb{Q}^T\right)}{d\lambda_T}(t) \right] - \log \left[ \frac{d\mathbb{P}^T}{d\lambda_T}(t) \right] \right) \nonumber \\
    &= \lim\limits_{\epsilon \to 0} \frac{1}{\epsilon} \log \left[ \frac{d\left(\mathbb{P}^T + \epsilon \mathbb{Q}^T\right)}{d\mathbb{P}^T}(t) \right] = \lim\limits_{\epsilon \to 0} \frac{1}{\epsilon} \log \left[ 1 + \epsilon \frac{d\mathbb{Q}^T}{d\mathbb{P}^T}(t) \right] \nonumber \\
    &= \lim\limits_{\epsilon \to 0} \frac{\frac{d\mathbb{Q}^T}{d\mathbb{P}^T}(t)}{1 + \epsilon \frac{d\mathbb{Q}^T}{d\mathbb{P}^T}(t)} = \frac{d\mathbb{Q}^T}{d\mathbb{P}^T}(t) = \frac{\frac{d\mathbb{Q}^T}{d\lambda_T}(t)}{\frac{d\mathbb{P}^T}{d\lambda_T}(t)} = \frac{q(t)}{p(t)} \; . \label{eqn_continuous_pid_differentiable_density_part_limit}
\end{align}
where we have used L'Hopitals rule for jointly undefined limits in the fifth equality. Here $q_T(t) = \frac{d\mathbb{Q}^T}{d\lambda_T}(t)$ is the marginal density of the deviation in probability measure, $\mathbb{Q}^T$, on $E_T$. Thus the derivative exists exactly when $p(t) \neq 0$ which was the same condition for $i_\cap^{\sx}(t:\alpha)$ to exist in the first place. 

We now use a violent abuse of notation; during the rest of the proof, a push forward will be carrying a new label connected to the map with respect to which we push forward. A measure with a superscript $\cdot^{\mathds{1}_{\alpha, s}}$ denotes a push forward by $\mathds{1}_{\alpha, s}$ while a $\cdot^V$, as usual, denotes a push forward by the joint random variable $V = (S_1, \ldots, S_{N_S}, T)$. 

Let again $\mathbb{Q}$ be chosen as described above. Combining \cref{eqn_regular_conditional_probability_as_rn_derivative} and the definition of the marginal measure given above, $\nu^T_{\alpha, s}$ is obtained as Radon-Nikod\'{y}m derivative by $\nu^T_{\alpha, s}(B) = \frac{d\mathbb{P}^{\mathds{1}_{\alpha, s}}( \mathds{1}_{\alpha, s} \circ \Pi_T(B) \cap \cdot)}{d\mathbb{P}^{\mathds{1}_{\alpha, s}}(\cdot)}(1)$ for all $ B \in \mathcal{B}\left(E_{T}\right)$. The argument $1$ here stems from the conditioning on $R_{\alpha, s}$, or equivalently, $\mathds{1}_{\alpha, s} = 1$.

Let $\tilde{\nu}_{\alpha, s, \mathbb{Q}, \epsilon}^T$ be the marginal regular conditional probability measure on $E_T$ that was obtained via $\mathbb{P} + \epsilon \mathbb{Q}$ instead of $\mathbb{P}$ as in \cref{thrm_disintegration}. Additionally, let $\mu_{\alpha, s, \mathbb{Q}}^T$ be the marginal regular conditional probability measure generated via $\mathbb{Q}$ alone, exactly as $\nu_{\alpha, s}^T$, just with $\mathbb{Q}$ instead of $\mathbb{P}$.  

Then investigating the partial differentiability of $\frac{d\nu_{\alpha, s}^T}{d\lambda_T}$ as a functional of $\mathbb{P}$ amounts to evaluate
\begin{align}
    \partial^{\mathbb{Q}}_{\mathbb{P}} \log \left[ \frac{d\nu_{\alpha, s}^T}{d\lambda_T} (t) \right] &= \lim\limits_{\epsilon \to 0} \frac{1}{\epsilon} \left( \log \left[ \frac{d\tilde{\nu}_{\alpha, s, \mathbb{Q}, \epsilon}^T}{d\lambda_T}(t) \right] - \log \left[ \frac{d\nu_{\alpha, s}^T}{d\lambda_T} (t) \right] \right) \nonumber \\
    &= \lim\limits_{\epsilon \to 0} \frac{1}{\epsilon} \log \left[ \frac{d\tilde{\nu}_{\alpha, s, \mathbb{Q}, \epsilon}^T}{d\nu_{\alpha, s}^T}(t) \right] \; . \label{eqn_continuous_pid_differentiable_regular_conditional_probability_part_limit}
\end{align}
Next we want to expand $\frac{d\tilde{\nu}_{\alpha, s, \mathbb{Q}, \epsilon}^T}{d\nu^T_{\alpha, s}}(t)$ into the Radon-Nikod\'{y}m derivatives with respect to the probability measures $\mathbb{P}$. 

Further, noticing that an expansion results in two consecutive Radon-Nikod\'{y}m derivatives in a single expression, we have to make clear how to understand the measures that are in the individual derivatives. To this end, we use the notation $\cdot_1$, $\cdot_2$, $\cdot_3$ and $\cdot_4$, denoting different sets, in order to distinguish the sets with respect to which we understand the measure to be varying inside the particular derivatives. More explicitly, if inside a Radon-Nikod\'{y}m derivative, the same set is present in the numerator as is in the denominator, say $\cdot_2$, the derivative is carried out as if $\cdot_1$ were constant, as in usual partial derivatives. Additionally, we believe this procedure of distinguishing sets amounts for enhanced readability, as one can now easily follow the calculation steps by simply following the sets $\cdot_i$,  $i=1,2,3,4$. Note that while $\cdot_1$ represents a set in $\mathcal{B}(E_T)$, $\cdot_{2, 3, 4} \in \mathcal{B}(\{0,1\})$.

Carrying out the expansion, we obtain
\begin{align*}
    \frac{d\tilde{\nu}_{\alpha, s, \mathbb{Q}, \epsilon}^T}{d\nu^T_{\alpha, s}}\left(t\right) &= \frac{d\tilde{\nu}_{\alpha, s, \mathbb{Q}, \epsilon}^T(\cdot_1)}{d\nu^T_{\alpha, s}(\cdot_1)}\left(t\right) = \frac{ d \left[\frac{ d\left(\mathbb{P}^{\mathds{1}_{\alpha, s}} + \epsilon \mathbb{Q}^{\mathds{1}_{\alpha, s}}\right)\left( \mathds{1}_{\alpha, s} \circ \Pi_T(\cdot_1) \cap \cdot_2\right) }{ d\left(\mathbb{P}^{\mathds{1}_{\alpha, s}} + \epsilon \mathbb{Q}^{\mathds{1}_{\alpha, s}}\right)(\cdot_2) } (1) \right]}{ d\nu^T_{\alpha, s}(\cdot_1) } \left(t\right) \\ 
    &= \frac{ d \left[\frac{ \frac{ d\left(\mathbb{P}^{\mathds{1}_{\alpha, s}} + \epsilon \mathbb{Q}^{\mathds{1}_{\alpha, s}}\right)\left( \mathds{1}_{\alpha, s} \circ \Pi_T(\cdot_1) \cap \cdot_2 \right)   }{d\mathbb{P}^{\mathds{1}_{\alpha, s}}(\cdot_2)} (1) }{ \frac{  d\left(\mathbb{P}^{\mathds{1}_{\alpha, s}} + \epsilon \mathbb{Q}^{\mathds{1}_{\alpha, s}}\right)(\cdot_4)}{d\mathbb{P}^{\mathds{1}_{\alpha, s}}(\cdot_4)} (1)}
    \right]}{ d\nu^T_{\alpha, s}(\cdot_1)} \left(t\right) \\
    &= \frac{ d \left[\frac{ \nu^T_{\alpha, s}(\cdot_1) + \epsilon  \frac{ d \mathbb{Q}^{\mathds{1}_{\alpha, s}}\left( \mathds{1}_{\alpha, s} \circ \Pi_T( \cdot_1 ) \cap \cdot_2 \right)  }{d\mathbb{P}^{\mathds{1}_{\alpha, s}}(\cdot_2)} (1) }{ 1 + \epsilon \frac{  d\mathbb{Q}^{\mathds{1}_{\alpha, s}}(\cdot_4)} {d\mathbb{P}^{\mathds{1}_{\alpha, s}}(\cdot_4)} (1)}
    \right]}{ d\nu^T_{\alpha, s}(\cdot_1)} \left(t\right) \; .
\end{align*}

Then the linearity of the Radon-Nikod\'{y}m derivative leads to
\begin{align*}
  \frac{d\tilde{\nu}_{\alpha, s, \mathbb{Q}, \epsilon}^T(\cdot_1)}{d\nu^T_{\alpha, s}(\cdot_1)}\left(t\right) &= \frac{ 1 }{ 1 + \epsilon \frac{  d\mathbb{Q}^{\mathds{1}_{\alpha, s}}(\cdot_4)} {d\mathbb{P}^{\mathds{1}_{\alpha, s}}(\cdot_4)} (1)} \left( \frac{ d \nu^T_{\alpha, s}(\cdot_1)
    }{ d \nu^T_{\alpha, s}(\cdot_1)} \left(t\right) + \epsilon \frac{ d \left[ \frac{ d\mathbb{Q}^{\mathds{1}_{\alpha, s}} \left( \mathds{1}_{\alpha, s} \circ \Pi_T( \cdot_1)  \cap \cdot_2\right) }{ d\mathbb{P}^{\mathds{1}_{\alpha, s}}(\cdot_2) } (1) 
    \right]}{ d\nu^T_{\alpha, s}(\cdot_1)} \left(t\right) \right) \\
    &= \frac{ 1 }{ 1 + \epsilon \frac{  d\mathbb{Q}^{\mathds{1}_{\alpha, s}}(\cdot_4)} {d\mathbb{P}^{\mathds{1}_{\alpha, s}}(\cdot_4)} (1)} \left( 1 + \epsilon \frac{d\mu_{\alpha, s, \mathbb{Q}}^T}{d\nu^T_{\alpha, s}}\left(t\right) \right) \; .
\end{align*}
With this finally the limit \cref{eqn_continuous_pid_differentiable_regular_conditional_probability_part_limit} can be evaluated as
\begin{align}
    \lim\limits_{\epsilon \to 0} \frac{1}{\epsilon} \log \left[ \frac{d\tilde{\nu}_{\alpha, s, \mathbb{Q}, \epsilon}^T(\cdot_1)}{d\nu^T_{\alpha, s}(\cdot_1)}\left(t\right) \right] &= \lim\limits_{\epsilon \to 0} \frac{1}{\epsilon} \left( \log \left[ 1 + \epsilon \frac{d\mu_{\alpha, s, \mathbb{Q}}^T}{d\nu^T_{\alpha, s}}\left(t\right) \right]  - \log \left[ 1 + \epsilon \frac{  d\mathbb{Q}^{\mathds{1}_{\alpha, s}}(\cdot_4)} {d\mathbb{P}^{\mathds{1}_{\alpha, s}}(\cdot_4)} (1) \right]  \right) \nonumber \\
    &\leftstackrel{\tiny L'Hopital}{=} \lim\limits_{\epsilon \to 0} \left( \frac{\frac{d\mu_{\alpha, s, \mathbb{Q}}^T}{d\nu^T_{\alpha, s}}\left(t\right)}{1 + \epsilon \frac{d\mu_{\alpha, s, \mathbb{Q}}^T}{d\nu^T_{\alpha, s}}\left(t\right) } - \frac{\frac{  d\mathbb{Q}^{\mathds{1}_{\alpha, s}}(\cdot_4)} {d\mathbb{P}^{\mathds{1}_{\alpha, s}}(\cdot_4)} (1)}{1 + \epsilon \frac{  d\mathbb{Q}^{\mathds{1}_{\alpha, s}}(\cdot_4)} {d\mathbb{P}^{\mathds{1}_{\alpha, s}}(\cdot_4)} (1)} \right) \nonumber \\
    &= \frac{d\mu_{\alpha, s, \mathbb{Q}}^T}{d\nu^T_{\alpha, s}}\left(t\right)
    - \frac{  d\mathbb{Q}^{\mathds{1}_{\alpha, s}}(\cdot_4)} {d\mathbb{P}^{\mathds{1}_{\alpha, s}}(\cdot_4)} (1) \; . \label{eqn_continuous_pid_differentiable_regular_conditional_probability_part_limit_evaluation}
\end{align}

Note that $\frac{d\mu_{\alpha, s, \mathbb{Q}}^T}{d\nu^T_{\alpha, s}}\left(t\right)$ can be expanded to be $\frac{ \frac{d\mu_{\alpha, s, \mathbb{Q}}^T}{d\lambda_T}(t)}{\frac{d\nu^T_{\alpha, s}}{\lambda_T}(t)}$, yielding the fraction of two conditional densities, $q_T(t | \alpha, s)$ and $p_T(t|\alpha, s)$.

Then the calculation here relates to absolute differentiability with respect to the probability density $p$ as can be seen when understanding $i^{\sx}_\cap(t : \alpha)[\mathbb{Q}]$ as a functional of $q$ instead of $\mathbb{Q}$. 

Combining \cref{eqn_continuous_pid_differentiable_density_part_limit} and \cref{eqn_continuous_pid_differentiable_regular_conditional_probability_part_limit_evaluation}, we see that the derivative of $i^{\sx}_\cap(t:\alpha)$ in the direction of $\mathbb{Q}$ becomes
\begin{align*}
    \partial^{\mathbb{Q}}_{\mathbb{P}} i^{\sx}_\cap (t:\alpha)[\mathbb{P}] = \frac{d\mu_{\alpha, s, \mathbb{Q}}^T}{d\nu^T_{\alpha, s}}(t)
    - \frac{  d\mathbb{Q}^{\mathds{1}_{\alpha, s}}} {d\mathbb{P}^{\mathds{1}_{\alpha, s}}} (1) - \frac{d\mathbb{Q}^T}{d\mathbb{P}^T}(t) \; .
\end{align*}
Writing this derivative in terms of the full density $p$ which corresponds to $\mathbb{P}$, that was perturbed in the direction of $q(s, t) = \frac{d\mathbb{Q}^V}{d\lambda}(s, t)$ with strength $\epsilon$ we obtain
\begin{align*}
    \partial^{q}_{p} i^{\sx}_\cap (t:\alpha)[p] =  \frac{q_T(t | \alpha, s)}{p_T(t | \alpha, s)} - \frac{ \int_{V\left(R_{\alpha, s}\right)} q_{\mathbf{S}}\left(s'\right) d\lambda_{\mathbf{S}}(s')} {\int_{V\left(R_{\alpha, s}\right)} p_{\mathbf{S}}\left(s'\right) d\lambda_{\mathbf{S}}(s')} - \frac{q_T(t)}{p_T(t)} \; .
\end{align*}
Here the densities $p_{\mathbf{S}}, q_{\mathbf{S}}$ correspond to the marginal distribution on $E_{-T}$. 
This form makes it evident that the shared exclusions based measure for redundant information is differentiable with respect to $p$ on the space of all integrable functions with respect to the measure $\lambda = \lambda_\mathbf{S} \times \lambda_T$.

\end{proof}

\begin{remark}
Moreover, since the resulting form is hyperbolic, and depends on both marginals on both $E_{T}$ and $E_{-T}$, it is continuous with respect to $p$, too. Further, we suspect $i^{\sx}_\cap[\mathbb{P}]$ to be smooth with respect to changes in the underlying probability measure. 
Thus we suspect that, as the function $\frac{1}{x}$ is smooth on $\mathbb{R} \setminus \{0\}$, the local shared exclusions' based redundant information $i_\cap^{\sx}$ is smooth whenever $p \neq 0$. Explicitly, this can be seen through again evaluating the derivative $\left. \big(\partial_\epsilon   \partial^q_p i^{\sx}_\cap(t:\alpha)[\mathbb{P} + \epsilon \mathbb{R}] \big) \right\vert_{\epsilon=0}$ for $\mathbb{R} \ll \mathbb{Q} \ll \mathbb{P}$ yet another probability measure, representing a directional derivative in an almost arbitrary other direction.
Evaluation leads to:
\begin{align}
    \left. \big(\partial_\epsilon   \partial^q_p i^{\sx}_\cap(t:\alpha)[\mathbb{P} + \epsilon \mathbb{R}] \big) \right\vert_{\epsilon=0}= &- \frac{d\mu_{\alpha, s, \mathbb{Q}}^T}{d\nu^T_{\alpha, s}}(t) \frac{d\mu_{\alpha, s, \mathbb{R}}^T}{d\nu^T_{\alpha, s}}(t)
    + \frac{  d\mathbb{Q}^{\mathds{1}_{\alpha, s}}} {d\mathbb{P}^{\mathds{1}_{\alpha, s}}} (1) \frac{  d\mathbb{R}^{\mathds{1}_{\alpha, s}}} {d\mathbb{P}^{\mathds{1}_{\alpha, s}}} (1) \nonumber \\
    &\quad \quad \quad + \frac{d\mathbb{Q}^T}{d\mathbb{P}^T}(t) \frac{d\mathbb{R}^T}{d\mathbb{P}^T}(t) \; .
\end{align}
We believe, this derivation can be continued arbitrarily many times, and believe that the functional $i^{\sx}_\cap[\mathbb{P}]$ is smooth with respect to the underlying probabilty measure $\mathbb{P}$.

\end{remark}

\subsubsection{Target chain rule for $i^{\sx}_\cap$} \label{subsubsection_properties_target_chain_rule}

For the quantity proposed in \cref{def_continuous_PID} to embed nicely in the information-theoretic framework of local mutual information, a treatment of composite targets needs to be considered. For that, we claim that the local shared exclusions' redundant information \cref{eqn_continuous_PID_collections} fulfills a chain rule for the target random variable $T$.

\begin{definition} \label{def_target_chain_rule_conditional_quantity}
Let $T = (T_1, T_2)$ be a joint target variable and fix a realization $t = (t_1, t_2)$. Then we define the local shared information of $t_2$ about $\alpha$, when conditioned on $t_1$, as
\begin{align} \label{eqn}
    i^{\sx}_\cap\left( t_2:\alpha \mid t_1 \right) := \log \left[ \frac{d\nu^{\mathds{1}_{\alpha, s}}_{t}}{d\nu^{\mathds{1}_{\alpha, s}}_{t_1}}(1)\right]
\end{align}
wherever $\nu_{t} \ll \nu_{t_1} \ll \mathbb{P}$.

Informally, this quantity arises by conditioning each of the measures in \cref{eqn_continuous_PID_collections} (when stated with respect to $t_2$) on $t_1$.
\end{definition}

\begin{proposition}[Target chain rule] \label{proposition_target_chain_rule_local_sxpid}
Let $T = (T_1, T_2)$ be a joint target variable and fix a realization $t = (t_1, t_2)$ as in \cref{def_target_chain_rule_conditional_quantity}. Then, if $\nu_{t} \ll \nu_{t_1} \ll \mathbb{P}$,
\begin{align}
    i^{\sx}_\cap(t:\alpha) = i^{\sx}_\cap(t_1 : \alpha) + i^{\sx}_\cap\left( t_2:\alpha \mid t_1 \right) \; .
\end{align}
\end{proposition}

\begin{proof}
Recall that according to \cref{eqn_monotonicity_bayes_reformulation}, $i^{\sx}_\cap(t:\alpha)$ is defined via $\log \left[ \frac{d\nu^{\mathds{1}_{\alpha, s}}_{t}}{d\mathbb{P}^{\mathds{1}_{\alpha, s}}}(1) \right]$.

Consider the measure $\nu^{\mathds{1}_{\alpha, s}}_{t_1}$ that is a well-defined regular measure using an auxiliary random variable $\mathds{1}_{t_1}$ similar as in \cref{def_auxiliary_random_variable}. Then
\begin{align*}
    \frac{d\nu^{\mathds{1}_{\alpha, s}}_{t}}{d\mathbb{P}^{\mathds{1}_{\alpha, s}}}(1) &= \frac{d\nu^{\mathds{1}_{\alpha, s}}_{t}}{d\nu^{\mathds{1}_{\alpha, s}}_{t_1}}(1) \frac{d\nu^{\mathds{1}_{\alpha, s}}_{t_1}}{d\mathbb{P}^{\mathds{1}_{\alpha, s}}}(1) \; ,
\end{align*}
where both Radon-Nikod\'{y}m derivatives exist since $\nu_{t} \ll \nu_{t_1} \ll \mathbb{P}$.

Therefore, when comparing with \cref{def_target_chain_rule_conditional_quantity},
\begin{align*}
    i^{\sx}_\cap(t = \{t_1, t_2\} : \alpha) &= \underbrace{ \log \left[  \frac{d\nu^{\mathds{1}_{\alpha, s}}_{t}}{d\nu^{\mathds{1}_{\alpha, s}}_{t_1}}(1) \right] }_{ = i^{\sx}_\cap\left( t_2:\alpha \mid t_1 \right)} + \underbrace{ \log \left[ \frac{d\nu^{\mathds{1}_{\alpha, s}}_{t_1}}{d\mathbb{P}^{\mathds{1}_{\alpha, s}}}(1)  \right]}_{ = i^{\sx}_\cap(t_1 : \alpha)} = i^{\sx}_\cap\left( t_2:\alpha \mid t_1 \right) + i^{\sx}_\cap(t_1 : \alpha) \; ,
\end{align*}
such that the conditioning on $t_1$ in the term $i^{\sx}_\cap\left( t_2:\alpha\mid t_1 \right)$ translates to conditioning both measures with respect to which the Radon-Nikod\'{y}m derivative is determined.
\end{proof}

\subsubsection{On the associated global measure $I^{\sx}_\cap$}

While the local redundancy-based partial information decomposition quantity $i^{\sx}_\cap(t:s)$ yields some insight about specific realizations and their relation to the local mutual information, for practical purposes a local quantity does not suffice. 
Therefore we introduce a global, realization-independent global redundant information $I^{\sx}_\cap(T:\alpha) = \int_{E} i^{\sx}_\cap(t:\alpha) d\mathbb{P}^V(s, t)$.

\begin{proposition}
Consider a decomposition of $I^{\sx}_\cap(T:\alpha)$ as in \cref{eqn_isxpid_decomposition}, that is 
\begin{align*}
    I^{\sx}_\cap(T:\alpha) = \hat{I}^{\sx +}_\cap(T:\alpha) - \hat{I}^{\sx -}_\cap(T:\alpha) \; ,
\end{align*}
where 
\begin{align*}
    \hat{I}^{\sx \pm}_\cap(T:\alpha) := \int_{E} \hat{i}^{\sx \pm}_\cap(t:\alpha) d\mathbb{P}^V(s, t) \; .
\end{align*}
Then the properties that have been derived in \cref{thrm_local_pid_axioms}, and  \cref{proposition_monotonicity,proposition_invariance_local_sxpid,proposition_differentiability_local_sxpid,proposition_target_chain_rule_local_sxpid}, also hold for $\hat{I}^{\sx \pm}_\cap$ and the full $I^{\sx}_\cap$, respectively.
\end{proposition}

\begin{proof}

The global measure $I^{\sx}_\cap$ naturally fulfills the criteria from \cref{thrm_local_pid_axioms}, as the expected value does not affect $\alpha$. Moreover, due to the monotonicity of the local quantity $\hat{i}^{\sx \pm}_\cap$, also $\hat{I}^{\sx \pm}_\cap$ fulfills 
\begin{align*}
    \hat{I}^{\sx \pm}_\cap(T:\alpha) = \int_{E} \hat{i}^{\sx \pm}_\cap(t:\alpha) d\mathbb{P}^V(s, t) &\leq \int_{E} \hat{i}^{\sx \pm}_\cap(t:\alpha;\mathbf{a}_{m+1}) d\mathbb{P}^V(s, t) = \hat{I}^{\sx \pm}_\cap(T:\alpha;\mathbf{a}_{m+1}) \; ,
\end{align*}
and is hence monotone as in \cref{proposition_monotonicity}.

Regarding the invariance, since both $i^{\sx}_\cap(t:\alpha)$ and $\mathbb{P}$ are invariant under transformations as discussed in \cref{proposition_invariance_local_sxpid}, in consequence $I^{\sx}_\cap(T:\alpha)$ is invariant under any topological isomorphism, too.

To show the differentiability with respect to $p$, as above in \cref{proposition_differentiability_local_sxpid}, we now utilize functional calculus. This is similar to a Fr\'{e}chet derivative used in calculus of variations.
The functional derivative of $I^{\sx}_\cap(T:\alpha)[p]$ as a functional of a density function $p$ in the direction $q$, denoted by $\frac{\Delta I^{\sx}_\cap(T:\alpha)} {\Delta p}$ is defined via the integral
\begin{align*}
    \int \frac{\Delta I^{\sx}_\cap(T:\alpha)} {\Delta p(s, t)} q(s, t) d\lambda(s,t) := \left. \left[ \frac{d}{d\epsilon} I^{\sx}_\cap(T:\alpha)[p+\epsilon q] \right] \right\vert_{\epsilon = 0} \; .
\end{align*}
Carrying out the derivative with respect to $\epsilon$, under the assumption that the integral is bounded, we find

\begin{align*}
   \frac{\Delta I^{\sx}_\cap(T:\alpha)} {\Delta p}(s, t) = i^{\sx}_\cap(t:\alpha)[p(s, t)] + p(s, t) \, \partial_p^q \, i^{\sx}_\cap(t:\alpha)[p(s, t)] \; .
\end{align*}
Because $i^{\sx}_\cap(t:\alpha)[p]$ was a differentiable functional of $p$ as demonstrated in \cref{proposition_differentiability_local_sxpid}, consequently, also $I^{\sx}_\cap$ is a differentiable functional with respect to changes in the underlying probability density function.

The linearity of the integral further assures the validity of the target chain rule as in \cref{proposition_target_chain_rule_local_sxpid}, i.e. 
\begin{align*}
    I^{\sx}_\cap(\{T_1, T_2\}:\alpha) &= \int_{E} i^{\sx}_\cap(\{t_1, t_2\}:\alpha) \; d\mathbb{P}^V(s, t_1, t_2) \\
    &= \int_{E} \left[ i^{\sx}_\cap\left( t_2:\alpha \mid t_1 \right) + i^{\sx}_\cap(t_1 : \alpha) \right] \; d\mathbb{P}^V(s, t_1, t_2) \\
    &= \int_{E} i^{\sx}_\cap\left( t_2:\alpha \mid t_1 \right) \; d\mathbb{P}^V(s, t_1, t_2) + \int_{E_{-T_2}} i^{\sx}_\cap(t_1 : \alpha) \; d\mathbb{P}^V_{-T_2}(s, t_1) \\
    &= I^{\sx}_\cap\left( T_2:\alpha \mid T_1 \right) + I^{\sx}_\cap(T_1 : \alpha) \; .
\end{align*}
Here $\mathbb{P}^V_{-T_2}$ denotes the marginal measure on $E_{-T_2}$.

\end{proof}

\section{Conclusion}

There have been manifold successful endeavors developing PID in the past years, both axiomatically and in terms of deriving quantities, defining a PID in discrete settings. This work has introduced a measure of redundant information, utilizing the more general framework of measure theory, which does not only recover the discrete case for the Dirac measure, but also enables usage in continuous, or discrete-continuously mixed sample spaces.

Moreover, this newly proposed quantity fulfills a number of properties desirable for application. Next to $i^{\sx}_\cap \, / \, I^{\sx}_\cap$ fulfilling a system of (local) axioms for being a sensible PID quantity, invariance under invertible topological isomorphisms makes it suitable for investigating real-world data sets, while differentiability along probability distributions enable studying processes on statistical manifolds, such as learning in neural networks, or statistical modelling. 

While the framework proposed here builds on mathematical rigor, there are still open questions, of practical nature, as well as conceptual ones, which we will investigate in future studies. 

The first issue is a fundamental one; in the discrete case the local $i^{\sx}_\cap$ admits a non-negative decomposition via the self-redundancy as in \cref{thrm_local_pid_axioms}, by resembling the local entropy and the local conditional entropy. The general measure-theoretic framework, however, does not immediately allow for such a non-negative decomposition, as for instance local differential entropy might take negative values. For the continuous case, this undermines the original operational interpretation of $i^{\sx \pm}_\cap$ from the discrete case as given by the authors in \cite{makkeh2021isx}.

However, this might be fixed by adding a term to each of the $i^{\sx \pm}_\cap$, such that they become positive; this term was denoted by $F$ in the footnote in \cref{subsubsection_self_redundancy_monotonicity_symmetry} and might be of the form $F(x) = \log\left[ \frac{d\mathbb{M}^V}{d\lambda}(x) \right]$ for an $\mathbb{M}$ an arbitrary $\sigma$-finite positive measure such that $\mathbb{P} \ll \mathbb{M}$. This idea is similar to what Jaynes has proposed in \cite{jaynes_probability_theory}. This ansatz for a solution would, however, introduce yet another parameter of an investigation, as choosing $\mathbb{M}$ is a vastly underdetermined problem.

Secondly, for our measure to applicable to finite real-world data sets, there needs to be a tractable estimator, derived from first principles. As far as our literature research has reached, there is no such estimator to this day. 
Our current ideas revolve around obtaining a limit via generation of (topological) nets from nearest neighbor searches. The almost sure continuity of the Radon-Nikod\'{y}m derivative as in \cref{eqn_continuous_PID_collections} will guarantee the convergence of an approximation of the associated finitely valued likelihood fractions.

\subsection*{Acknowledgments}
We  would  like  to  thank  Raul Vicente, Juhan Aru and David Ehrlich for discussions and their valuable comments on this paper. MW received support from SFB Project No. 1193, Sub-project No. C04 funded by the Deutsche Forschungsgemeinschaft. MW, AM, KSP and AG are employed at the Campus Institute for Dynamics of Biological Networks(CIDBN) funded by the Volkswagen Stiftung. MW and AM received support from the Volkswagen Stiftung under the  program  “Big  Data  in  den  Lebenswissenschaften”. KSP received financial support from Honda Research Institute Europe. This work was supported by a funding from the Ministry for Science and Education of Lower Saxony and the Volkswagen Foundation through the “Niedersächsisches Vorab.” 

\newpage

\bibliographystyle{ieeetr}

\end{document}

%% file: special_commands
\DeclareMathOperator{\sx}{sx}

\newlength{\leftstackrelawd}
\newlength{\leftstackrelbwd}
\def\leftstackrel#1#2{\settowidth{\leftstackrelawd}%
{${{}^{#1}}$}\settowidth{\leftstackrelbwd}{$#2$}%
\addtolength{\leftstackrelawd}{-\leftstackrelbwd}%
\leavevmode\ifthenelse{\lengthtest{\leftstackrelawd>0pt}}%
{\kern-.5\leftstackrelawd}{}\mathrel{\mathop{#2}\limits^{#1}}}

\newtheorem{theorem}{Theorem}[section]
\newtheorem{proposition}{Proposition}[section]
\newtheorem{corollary}{Corollary}[section]

\newtheorem{definition}{Definition}[section]

\newtheorem{remark}{Remark}[section]